%% file: main.tex
\documentclass{article}
\input{modules/packages}
\input{modules/macros}

\input{modules/tikzstyles}

\usepackage{microtype}

\begin{document}
\bibliographystyle{eptcs}

\title{Operational Mermin non-locality \\ and All-vs-Nothing arguments}

\author{Stefano Gogioso \\ 
    Quantum Group, Department of Computer Science\\
    University of Oxford, UK\\
    \texttt{stefano.gogioso@cs.ox.ac.uk} 
}

\maketitle

\begin{abstract}
\noindent Contextuality is a key resource in quantum information and the device-independent security of quantum algorithms. In this work, we show that the recently developed, operational Mermin non-locality arguments of \cite{StefanoGogioso-MerminNonLocality} provide a large, novel family of quantum realisable All-vs-Nothing models \cite{NLC-AvNSheaf}. In particular, they result in a diverse wealth of quantum realisable models which are maximally contextual (i.e. lie on the faces of the no-signalling polytope with no local elements), and could be used as a resource for the security of a new class of quantum secret sharing algorithms.
\end{abstract}

\setcounter{tocdepth}{2}

\section{Introduction}

	Ever since Bell's original work \cite{QTC-bell}, contextuality has evolved from spooky phenomenon to fundamental feature of quantum mechanics, with applications to device-independent quantum security \cite{QTC-NoSignallingQKD}\cite{QTC-ContextualityDeviceIndepSec} and a recently proposed role in quantum speed-up \cite{QTC-howard2014contextuality}. Contrary to non-locality, which has many inequivalent definitions in the different communities, contextuality comes with a reasonably standard definition in terms of measurement contexts and probability distributions on outcomes (known as empirical models), and is rigorously captured by the sheaf-theoretic framework introduced in \cite{NLC-SheafSeminal}. 

	Of the many non-locality arguments that followed Bell's, Mermin's non-locality argument \cite{QTC-mermin1990quantum} stands out for its elegance and simplicity, and its $N$-partite generalisations can be directly translated into a family quantum secret sharing protocols known as HBB CQ \cite{QTC-HBB}\cite{QTC-HBB2}. Mermin's original non-locality argument is based on a system of equations, each admitting a solution in $\integersMod{2}$ but without a global solution: the probabilities don't play any role in the argument, which admits a purely possibilistic treatment. The formulation in terms of equations directly implies a much stronger form of contextuality, and can be generalised to a large class of possibilistic contextuality arguments known as All-vs-Nothing models \cite{NLC-AvNSheaf}. There is considerable interest in quantum realisable All-vs-Nothing models (like the one from Mermin's original argument) because all such models would automatically be maximally contextual, lying on a face of the no-signalling polytope.

	The traditional linear-algebraic formulation of quantum mechanics makes it hard to isolate and understand the operational building blocks that lead to quantum advantage in quantum information and computation, as well as non-classicality in quantum foundations. The framework of Categorical Quantum Mechanics \cite{CQM-seminal} has been developed throughout the years to provide a concrete, hands-on language that describes many fundamental structures involved in the theory and applications of quantum mechanics. Mermin's original non-locality argument was formalised in this language by \cite{CQM-StrongComplementarity}, unearthing a novel connection between contextuality in Mermin's argument and the structure of phase groups in quantum mechanics, and a treatment of the HBB CQ protocols appears in \cite{QTC-HBBVlad}.

	A complete characterisation of Mermin non-locality in terms of phase groups recently appeared in \cite{StefanoGogioso-MerminNonLocality}, leading to a large class of contextuality arguments generalising Mermin's original argument. Instead of focusing on the possibilistic distribution of outcomes and the system of locally-solvable/globally-unsolvable equations, this new approach focuses on the operational aspects, involving phases and eigenstates of the Pauli observables. In particular, it generalises the single equation $2y = 1$, with no solution in $\integersMod{2}$, which is used in Mermin's original argument to prove the non-existence of global solutions for the system of equations.\\

	\noindent In Section \ref{section_OperationalMerminNonLocality}, we provide an alternative and more discursive presentation of the material in \cite{StefanoGogioso-MerminNonLocality}, and we show that all the \inlineQuote{operational Mermin non-locality arguments} described therein are quantum realisable. 

	In Section \ref{section_AllvsNothingArguments}, we draw the connection with the sheaf-theoretic framework and All-vs-Nothing models, and we show that the operational Mermin non-locality arguments provide a new infinite family of quantum realisable All-vs-Nothing models. Furthermore, we show how operational Mermin non-locality arguments can be used to provide a non-collapsing hierarchy of All-vs-Nothing models requiring arbitrary large finite fields for their formulation.


\section{Operational Mermin non-locality}
	\label{section_OperationalMerminNonLocality}
	
	\noindent The first part of this paper presents the work of \cite{StefanoGogioso-MerminNonLocality} on Mermin non-locality\footnote{In this work, the word non-locality is used in \inlineQuote{Mermin non-locality} for historical reasons, but in all technical contexts we will prefer the word contextuality, to take away any residual emphasis on underlying space-time structure carried by the expression \inlineQuote{non-locality}.} in a format more easily accessible to the quantum information community, and provides a novel result on quantum realisability. Mermin's original non-locality argument is summarised, with a particular focus on the role played by phases. Finite-dimensional Hilbert spaces are generalised to finite-dimensional free modules over involutive semirings: GHZ states, phase gates and measurements/decoherence are introduced in this new context. Finally, Mermin's original non-locality argument is fully generalised to obtain a large family of quantum realisable non-locality arguments: because of the focus on concrete realisation in $\dagger$-symmetric monoidal categories, we shall refer to this more general family as the \textbf{operational Mermin non-locality arguments}.

	\subsection{Mermin's original non-locality argument}

		\noindent In the original \cite{QTC-mermin1990quantum}, Mermin considers a 3-qubit GHZ state in the computational basis, the basis of eigenstates of the single-qubit Pauli $Z$ observable, together with the following 4 measurement contexts:
		\begin{enumerate}
			\item[(a)] The GHZ state is measured in the observable $X_1 \tensor X_2 \tensor X_3$.\footnote{Where $X_j$ is the single-qubit Pauli $X$ observable on qubits $j=1,2,3$.} 
			\item[(b)] The GHZ state is measured in the observable $Y_1 \tensor Y_2 \tensor X_3$.\footnote{Where $Y_j$ is the single-qubit Pauli $Y$ observable on qubits $j=1,2,3$.} 
			\item[(c)] The GHZ state is measured in the observable $Y_1 \tensor X_2 \tensor Y_3$.
			\item[(d)] The GHZ state is measured in the observable $X_1 \tensor Y_2 \tensor Y_3$.
		\end{enumerate}

		\noindent Following traditional notation, we denote by $\ket{z_0},\ket{z_1}$ the eigenstates of the single-qubit Pauli $Z$ observable, by $\ket{\pm} := \ket{z_0} \pm \ket{z_1}$ those of the single-qubit Pauli $X$ observable and by $\ket{\pm i} := \ket{z_0} \pm i \ket{z_1}$ those of the single-qubit Pauli $Y$ observable. We can see measurement outcomes as valued in $\integersMod{2}$ by fixing the following bijections:
		\begin{enumerate}
			\item[(i)] for the $X$ observable, $\ket{+} \mapsto 0$ and $\ket{-} \mapsto 1$
			\item[(ii)] for the $Y$ observable, $\ket{+i} \mapsto 0$ and $\ket{-i} \mapsto 1$
		\end{enumerate}
 
 		\noindent Mermin argument then proceeds as follows. While the joint measurement outcomes are probabilistic, the $\integersMod{2}$ sum of the outcomes turns out to be deterministic, yielding the following system of equations:
 		\begin{equation}
 		\label{eqn_MerminSystemZ2Equations}
			\begin{cases}
			X_1 \oplus X_2 \oplus X_3 &= 0 \\
			Y_1 \oplus Y_2 \oplus X_3 &= 1 \\
			Y_1 \oplus X_2 \oplus Y_3 &= 1 \\
			X_1 \oplus Y_2 \oplus Y_3 &= 1 
			\end{cases}
		\end{equation}

		\noindent If there was a non-contextual assignment of outcomes for all measurements ($X_1,X_2,X_3,Y_1,Y_2$ and $Y_3$), i.e. if there existed a non-contextual hidden variable model, then the system of equations \ref{eqn_MerminSystemZ2Equations} would have a solution in $\integersMod{2}$, and in particular it would have to be consistent. However, the sum of the left hand sides yields $0$ in $\integersMod{2}$:
		\begin{equation}
			\label{eqn_MerminSystemZ2EquationsLHSSum}
			2 X_1 \oplus 2 X_2 \oplus ... \oplus 2 Y_3 = 0 X_1 \oplus ... \oplus 0 Y_3 = 0
		\end{equation}
		\noindent while the sum of the right hand sides yields $0 \oplus 1 \oplus 1 \oplus 1 = 3 = 1$ in $\integersMod{2}$. This shows the system to be inconsistent. Equivalently, one could observe that the sum of the LHS from \ref{eqn_MerminSystemZ2EquationsLHSSum} can equivalently be written as $2(Y_1 \oplus Y_2 \oplus Y_3)$, and that the inconsistency of the system of equations is witnessed by the fact that the equation $2 y = 1$ has no solution in $\integersMod{2}$. This latter point of view is the key to the operational generalisation of Mermin non-locality, while the All-vs-Nothing generalisation has its focus on inconsistent systems of equations.

	\subsection{The role of phases in Mermin's argument}

		\newcommand{\phasegate}[1]{P_{#1}}

		\noindent To understand the role played by the equation $2 y = 1$ in the original Mermin argument, we take a step back. First of all, we observe that the single-qubit Pauli $Y$ measurement can be equivalently obtained as a single-qubit Pauli $X$ measurement preceded by an appropriate unitary. A single-qubit \textbf{phase gate}, in the computational basis (single-qubit Pauli $Z$ observable), is a unitary transformation in the following form:
		\begin{equation}
			\label{eqn_Z2PhaseGates}
			\phasegate{\alpha} := \left(
			\begin{array}{cc}
				1 & 0 \\
				0 & e^{i \alpha}
			\end{array}
			\right)
		\end{equation}
		\noindent where we used the fact that global phases are irrelevant to set the first diagonal element to 1. Then measuring in the single-qubit $Y$ observable is equivalent to first applying the single-qubit phase gate $\phasegate{\frac{\pi}{2}}$ and then measuring in the single-qubit Pauli $X$ observable.

		Because they pairwise commute, phase gates come with a natural abelian group structure given by composition, resulting in an isomorphism $\alpha \mapsto P_\alpha$ between them and the abelian group $\reals/(2\pi\integers)$ (isomorphic to the circle group $S^1$). Of all the phase gates, $\phasegate{0}$ (the identity element of the group) and $\phasegate{\pi}$ stand out because of their well-defined action on the eigenstates of the single-qubit Pauli $X$ observable:
		\begin{align}
			\label{eqn_PzeroPiAction}
			\phasegate{0} & = \ket{\pm} \mapsto \ket{\pm} \nonumber\\
			\phasegate{\pi} & = \ket{\pm} \mapsto \ket{\mp} 
		\end{align}
		\noindent If we see $\ket{\pm}$ as the subgroup $\{0,\pi\} < \reals / (2\pi\integers)$ (corresponding to $\{\pm 1\} < S^1$ in the circle group), then Equation \ref{eqn_PzeroPiAction} looks a lot like the regular action of $\{0,\pi\}$ on itself. This is not a coincidence. Each phase gate $\phasegate{\alpha}$ can be (faithfully) associated the unique \textbf{phase state} $\ket{\alpha} := \ket{z_0} + e^{i \alpha} \ket{z_1}$ obtained from its diagonal, and these phase states can be abstractly characterised in terms of the single-qubit Pauli $Z$ observable, with no reference to phase gates (see the next section for the characterisation of phase states). The phase states inherit the abelian group structure of the phase gates, and their regular action coincides with the action of the group of phase gates on them. In particular, the phase gates $\phasegate{0}$ and $\phasegate{\pi}$ have the eigenstates of the single-qubit Pauli $X$ observable as their associated phase states $\ket{0} = \ket{+}$ and $\ket{\pi} = \ket{-}$, endowing the outcomes of single-qubit Pauli $X$ measurements with the natural $\integersMod{2}$ abelian group structure arising\footnote{There is a unique isomorphism $\integersMod{2} \isom \{0,\pi\}$.} from the inclusion $\{0,\pi\} < \reals/(2\pi\integers)$. We will refer to the group of phase states as the \textbf{group of $Z$-phases}, and to the subgroup $\{0,\pi\}$ as the \textbf{subgroup of $X$-classical points}, which we will also use to label the corresponding measurement outcomes for the single-qubit Pauli $X$ observable. We now show how to re-construct Mermin's argument from the following statement: the equation $2 y = \pi$ has no solution in the subgroup $\{0,\pi\} < \reals/(2\pi\integers)$ of $X$-classical points, but has a solution $y = \frac{\pi}{2}$ (corresponding to $y = e^{i\frac{\pi}{2}} = +i$ in the circle group) in the group $\reals/(2\pi\integers)$ of $Z$-phases.\\

		\noindent The GHZ state used in Mermin's argument has a special property, due to strong complementarity, when it comes to phase gates followed by measurements in the single-qubit Pauli $X$ observable.\cite{CQM-StrongComplementarity}

		\begin{lemma} If $\alpha_j \in \reals/(2\pi\integers)$, denote by $X_j^{\alpha_j}$ the measurement (outcome) on qubit $j$ obtained by first applying phase gate $\phasegate{\alpha_j}$ and then measuring in the single-qubit Pauli $X$ observable. If $\alpha_1 \oplus \alpha_2 \oplus \alpha_3 = \modclass{0\text{ or }\pi}{2\pi}$, then $X_1^{\alpha_1} \oplus X_2^{\alpha_2} \oplus X_3^{\alpha_3} = \modclass{0 \text{ or }\pi}{2\pi}$ respectively.
		\end{lemma}

		\noindent In the particular case of $X_j := X_j^{0}$ and $Y_j := X_j^{\frac{\pi}{2}}$, we obtain the system of equations from \ref{eqn_MerminSystemZ2Equations}, where now $\oplus$ is the sum in the abelian group $\{0,\pi\} < \reals/(2\pi\integers)$, instead of the original $\integersMod{2}$:
		 \begin{equation}
 		\label{eqn_MerminSystemZ2EquationsPM}
			\begin{cases}
			X_1 \oplus X_2 \oplus X_3 &= 0 \\
			Y_1 \oplus Y_2 \oplus X_3 &= \pi \\
			Y_1 \oplus X_2 \oplus Y_3 &= \pi \\
			X_1 \oplus Y_2 \oplus Y_3 &= \pi 
			\end{cases}
		\end{equation}

		\noindent Now back to the equation $2 y = \pi$, which has solution $y = \frac{\pi}{2}$ in $\reals/(2\pi\integers)$, but no solution in $\{0,\pi\}$. Consider an $N$-partite GHZ state (with $N \geq 2$), the measurement $X_1^{\frac{\pi}{2}} \otimes X_2^{\frac{\pi}{2}} \otimes X_3^{0} \otimes ... \otimes X_N^{0}$ and its $N-1$ non-trivial cyclic permutations. This yields the following generalised system of equations, where all the right hand sides are $\pi$ because we chose phase gates based on the solution $y = \frac{\pi}{2}$ in $\reals/(2\pi\integers)$ to the equation $2 y = \pi$:
		\begin{equation}
			\begin{cases}
			X_1^{\frac{\pi}{2}} \oplus X_2^{\frac{\pi}{2}} \oplus X_3^{0} \oplus ... \oplus X_{N-1}^{0} \oplus X_N^{0} &= \pi \\
			X_1^{\frac{\pi}{2}} \oplus X_2^{0} \oplus ... \oplus X_{N-2}^{0} \oplus X_{N-1}^{0} \oplus X_N^{\frac{\pi}{2}} &= \pi \\
			\hspace{2.5cm}\vdots \\
			X_1^{0} \oplus X_2^{\frac{\pi}{2}} \oplus X_3^{\frac{\pi}{2}} \oplus X_4^{0} \oplus ... \oplus X_{N}^{0} &= \pi \\
			\end{cases}
		\end{equation}

		\noindent Adding up (in the abelian group $\reals/(2\pi\integers)$) all left hand sides gives the following equation:
		\begin{equation}
		\label{eqn_DraftMerminEquation}
		(N-2) \left( X_1^{0} \oplus ... \oplus X_N^{0} \right) + 2y = N \pi
		\end{equation}
		where we defined $y := X_1^{\frac{\pi}{2}} \oplus X_2^{\frac{\pi}{2}} \oplus ... \oplus X_N^{\frac{\pi}{2}} $. Taking $N = \modclass{1}{k}$, where $k = 2$ is the exponent of the group $\integersMod{2}$, makes the right hand side of \ref{eqn_DraftMerminEquation} into $N \pi = \pi$; the smallest such $N \geq 2$ is $N=3$, yielding a 3-partite GHZ state. Adding the left hand side of Equation \ref{eqn_DraftMerminEquation} to $\left(k-(N-2)\right) \left( X_1^{0} \oplus ... \oplus X_N^{0} \right)$ and the right hand side to $\left(k-(N-2)\right) 0$ leaves us with the equation $2y=\pi$. 

		Then the following system of equations, the same system from \ref{eqn_MerminSystemZ2Equations} but with phase gate notation, can be seen to be inconsistent by adding up the three variations as $\reals/(2\pi\integers)$ equations, then adding $2-(3-2) = 1$ (an integer) times the control (an $\reals/(2\pi\integers)$ equation) and obtaining the $\reals/(2\pi\integers)$ equation $2 y = \pi$, which has no solution in the subgroup $\{0,\pi\}$ of $X$-classical points and thus excludes non-contextual hidden variable models:
		 \begin{equation}
 		\label{eqn_MerminSystemZ2EquationsPMGeneralisedPhases}
			\begin{cases}
			X_1^{0} \oplus X_2^{0} \oplus X_3^{0} = 0 \text{, the control}\\
			X_1^{\frac{\pi}{2}} \oplus X_2^{\frac{\pi}{2}} \oplus X_3^{0} = \pi \text{, the first variation} \\
			X_1^{\frac{\pi}{2}} \oplus X_2^{0} \oplus X_3^{\frac{\pi}{2}} = \pi \text{, the second variation} \\
			X_1^{0} \oplus X_2^{\frac{\pi}{2}} \oplus X_3^{\frac{\pi}{2}} = \pi \text{, the third variation} 
			\end{cases}
		\end{equation}		
		
		\noindent We just saw how phase gates, with their role in measurements of the GHZ state, allowed us to reconstruct the Mermin argument from the fact that the equation $2 y = \pi$ has solutions in the group of $Z$-phases, allowing the argument to be formulated, but not in the subgroup $X$-classical points, disallowing the existence of a non-contextual hidden variable model. In the next section we generalise this technique to arbitrary pairs of strongly complementary observables (generalising single-qubit Pauli $Z$ and Pauli $X$), in arbitrary $\dagger$-symmetric monoidal categories (henceforth $\dagger$-SMCs, generalising finite-dimensional Hilbert spaces).

	\subsection{From Hilbert spaces to modules of semirings}

		In \cite{CQM-StrongComplementarity} the original Mermin non-locality argument from the previous section is formalised in the context of Categorical Quantum Mechanics by using strong complementarity. In \cite{StefanoGogioso-MerminNonLocality}, the argument is fully generalised, and a completely algebraic characterisation of Mermin non-locality, valid in arbitrary $\dagger$-SMCs, is provided. Here we present the work of \cite{StefanoGogioso-MerminNonLocality} in a language closer to the one traditionally used in the study of quantum information.

		Instead of the field $\complexs$ of complex numbers, equipped with the involution given by complex conjugation, we will consider the more general case of an involutive commutative semiring $R$.\footnote{We require $0 \neq 1$ in $R$.} We will substitute finite-dimensional Hilbert spaces and linear maps with finite-dimensional free $R$-modules (henceforth \textbf{spaces}) and $R$-linear maps (henceforth \textbf{morphisms}). We will refer to this as a \textbf{process theory} (with superposition).\footnote{Which is easier on the tongue than \inlineQuote{dagger symmetric monoidal category distributively enriched in commutative monoids}.} 

		Given a basis $\ket{x}_{x \in X}$ of states of a space $\SpaceH$\footnote{I.e. elements of the $R$-modules, corresponding to vectors in a vector space. We will equivalently see a state of an $R$-module $\SpaceH$ as the unique morphism $R \rightarrow \SpaceH$ given by $r \mapsto r \ket{\psi}$.}, every state $\ket{\psi}$ of $\SpaceH$ can be written as follows, for a unique family $(\psi_x)_{x \in X}$ of coefficients in $R$:
		\begin{equation}
			\ket{\psi} = \sum_{x \in X} \psi_x \ket{x}
		\end{equation}
		
		\noindent We have a $\dagger$ on states given as follows, where $\ket{x}_{x \in X}$ is any orthonormal basis ($*:R \rightarrow R$ is the involution):
		\begin{equation}
			\bra{\psi} \text{ is defined to be the map } \SpaceH \rightarrow R \text{ sending any } \ket{\varphi} \text{ to } \sum_{x \in X} \psi_x^\star \varphi_x 
		\end{equation}
		
		\noindent More in general, given orthonormal bases $\ket{x}_{x \in X}$ and $\ket{y}_{y \in Y}$ of two free $R$-modules $\SpaceH$ and $\SpaceG$ respectively, any $R$-linear map $U$ can be written as follows, for a unique family $(U^x_y)_{x \in X,y \in Y}$ of coefficients in $R$:
		\begin{equation}
			U = \sum_{x \in X}\sum_{y \in Y} \ket{y} U^x_y\bra{x}
		\end{equation}
		
		\noindent We have a $\dagger$ on $R$-linear maps given by:
		\begin{equation}
			U^\dagger = \sum_{x \in X}\sum_{y \in Y} \ket{x} (U^x_y)^\star\bra{y}
		\end{equation}

		\noindent Finally, the tensor product $\tensor$ sends free $R$-modules with bases $\ket{x}_{x \in X}$ and $\ket{y}_{y \in Y}$ to the free $R$-module over the basis $\ket{x} \tensor \ket{y}_{x \in X, y\in Y}$.

		\begin{remark}
		From the point of view of \cite{StefanoGogioso-FourierTransform}, this is a $\dagger$-SMC distributively enriched over commutative monoids, where all objects admit some classical structure with enough classical points. In this context, classical structures with enough classical points (and such that the classical points form a finite, normalisable family) always correspond to orthonormal bases. We shall use this language no more for the rest of this paper.
		\end{remark}

	\subsection{GHZ states}

		\noindent From now on we fix some arbitrary space $\SpaceH$ and work in a finite orthonormal basis $\ket{x}_{x \in X}$, which we shall refer to as the \textbf{$X$ observable}, or the \textbf{$X$ basis}. We will write $d$ for the cardinality of $X$. Furthermore, suppose that for some abelian group structure $(X,\oplus,0)$ on $X$ there are morphisms $\hbox{\input{modules/symbols/timematchSym.tex}}\!:\SpaceH \tensor \SpaceH \rightarrow \SpaceH$ and $\hbox{\input{modules/symbols/timematchunitSym.tex}}\!:R \rightarrow \SpaceH$ given by:
		\begin{align}
			\hbox{\input{modules/symbols/timematchSym.tex}}\! & = \sum_{x,x'} \ket{x\oplus x'} \tensor \bra{x} \tensor \bra{x'} \\
			\hbox{\input{modules/symbols/timematchunitSym.tex}}\! & = \ket{0}
		\end{align}

		\noindent The internal monoid $(\hbox{\input{modules/symbols/timematchSym.tex}}\!,\hbox{\input{modules/symbols/timematchunitSym.tex}}\!)$, together with its adjoint, forms what is known as a quasi-special commutative $\dagger$-Frobenius algebra \cite{StefanoGogioso-FourierTransform}, which we shall refer to as the \textbf{$Z$ observable}.  Although it is not necessarily true that this algebra will have classical points forming a basis, and thus that it can be given the same interpretation as non-degenerate observables in quantum mechanics, we adopt this nomenclature to highlight the fact that the associated GHZ state will play the same role that was played in Mermin's original argument by the GHZ state in the single-qubit Pauli $Z$ basis. As a technical requirement, we will ask for the natural number $d$ to have a multiplicative inverse $d^{-1}$ as an element of the semiring $R$.

		The adjoint $\hbox{\input{modules/symbols/timediagSym.tex}}\!: \SpaceH \rightarrow \SpaceH \tensor \SpaceH$ of the morphism $\hbox{\input{modules/symbols/timematchSym.tex}}\!$ is given as follows, and will be used to construct the GHZ state:
		\begin{equation}
			\hbox{\input{modules/symbols/timediagSym.tex}}\! = \sum_{x} \sum_{x' \oplus x'' = x} \ket{x'} \tensor\ket{x''} \tensor \bra{x}
		\end{equation}
		By composing\footnote{But not tensoring.} together $N-1$ copies of $\hbox{\input{modules/symbols/timediagSym.tex}}\!$, we can obtain as many different morphisms $\SpaceH \rightarrow \SpaceH^{\tensor N}$ as there are binary trees with $N-1$ nodes. However, the group addition $\oplus$ is associative, and with it $\hbox{\input{modules/symbols/timematchSym.tex}}\!$ and $\hbox{\input{modules/symbols/timediagSym.tex}}\!$: hence all the morphisms $\SpaceH \rightarrow \SpaceH^{\tensor N}$ that can be obtained from $\hbox{\input{modules/symbols/timediagSym.tex}}\!$ (by composition only) coincide with the following morphism:

		\newcommand{\GeneralisedGHZ}[2]{\operatorname{GHZ}^{#1}_{#2}}
		\newcommand{\GHZ}[1]{\operatorname{GHZ}^{#1}}
		\begin{equation}
			\sum_{x} \ket{\GeneralisedGHZ{N}{x}} \bra{x}
		\end{equation}
		where the \textbf{$N$-partite generalised GHZ state} $\ket{\GeneralisedGHZ{N}{x}}$ (with respect to the $Z$ observable) is given by: 
		\begin{equation}
		\label{eqn_GeneralisedGHZ}
			\ket{\GeneralisedGHZ{N}{x}} := \sum_{x_1\oplus...\oplus x_N = x} \ket{x_1} \tensor ... \tensor \ket{x_N}
		\end{equation}

		\noindent The \textbf{$N$-partite GHZ state} $\ket{\GHZ{N}}$ is defined to be the generalised GHZ state at the group element $x=0$:  
		\begin{equation}
		\label{eqn_GHZ}
			\ket{\GHZ{N}} := \ket{\GeneralisedGHZ{N}{0}} = \sum_{x_1\oplus...\oplus x_N = 0} \ket{x_1} \tensor ... \tensor \ket{x_N}
		\end{equation}		

		\noindent The state in Equation \ref{eqn_GHZ} is expressed in terms of the $X$ observable, while the GHZ state is traditionally written in the $Z$ observable: how is this new expression related to the traditional one? In the case of finite-dimensional Hilbert spaces\footnote{But this can be done in more generality.}, the orthogonal basis $\ket{z}_{z \in Z}$ associated with our $Z$ observable would take the following form, where $Z$ is the set (abelian group $(Z,\cdot,1)$, in fact) of multiplicative characters $z:X \rightarrow S^1$ of the abelian group $(X,\oplus,0)$: 
		\begin{equation}
		\label{eqn_ZbasisHilb}
			\ket{z} := \sum_{x} z(x)^\star \ket{x}
		\end{equation}
		By the fundamental theorem of finite abelian groups, we can always write $X = \prod_{j \in J}\integersMod{n_j}$ for some natural numbers (in fact, prime powers) $n_j$: in this case, elements of $X$ can be written as $J$-indexed vectors, with the $j$-th component valued in $\integersMod{n_j}$, and Equation \ref{eqn_ZbasisHilb} takes the more familiar form:		
		\begin{equation}
		\label{eqn_ZbasisHilbExplicit}
			\ket{z_{\underline{y}}} := \sum_{\underline{x}} \exp\left[- \sum_{j} \frac{y_j x_j}{n_j}\right] \ket{\underline{x}}
		\end{equation}
		where we have fixed some isomorphism\footnote{We can because every finite group is isomorphic to its Pontryagin dual, but our choice of iso is, in general, non-canonical.} $(X,\oplus,0) \isom (Z,\cdot,1)$, bijectively sending $\underline{y} \in X$ to $z_{\underline{y}} \in Y$. Equation \ref{eqn_ZbasisHilbExplicit} can be inverted to obtain write the $X$ basis in terms of the $Z$ basis:
		\begin{equation}
		\label{eqn_ZbasisHilbExplicitInverted}
			\ket{\underline{x}} := \sum_{\underline{y}} \exp\left[\sum_{j} \frac{y_j x_j}{n_j}\right] \ket{z_{\underline{y}}}
		\end{equation}

		\noindent Then Equation \ref{eqn_GHZ} can be written as follows in terms of the $Z$ basis, recovering the traditional definition of GHZ state (generalised from $Z_2$ to an arbitrary finite abelian group $(X,\oplus,0)$):
		\begin{align}
		\label{eqn_GHZTraditional}
			\ket{\GHZ{N}} & = \sum_{\underline{y}_1}...\sum_{\underline{y}_N} \left[\sum_{\underline{x}_1+...+\underline{x}_N = 0}\exp\left[ \sum_{j \in J} \sum_{i=1}^{N} \frac{y_{ij}x_{ij}}{n_j} \right] \right] \ket{z_{\underline{y}_1}}  \tensor ... \tensor \ket{z_{\underline{y}_N}} = \nonumber \\
			& \propto \sum_{\underline{y}} \ket{z_{\underline{y}}} \tensor .... \tensor \ket{z_{\underline{y}}}
		\end{align}	
		where we have used the fact that the sum of exponential in square brackets evaluates to $d^{N-1}$ if $\underline{y}_1 = ... = \underline{y}_N$, and vanishes otherwise. In the $(X,\oplus,0) \isom \integersMod{2}$ case of single qubits, we recover the usual formulation of the $N$-partite GHZ state: 
		\begin{equation}
		\GHZ{N} \propto \ket{z_0}^{\tensor N} + \ket{z_1}^{\tensor N} \text{ in the single-qubit case}
		\end{equation}

	\subsection{Mermin measurement contexts}

		\noindent We now define phase gates for the $Z$ observable, generalising those of Equation \ref{eqn_Z2PhaseGates}. A \textbf{phase state} for the $Z$ observable is a state $\ket{\psi}$ such that the following holds:
		\begin{equation}
		(\bra{\psi} \tensor \id{\SpaceH} ) \cdot \hbox{\input{modules/symbols/timediagSym.tex}}\! \cdot \ket{\psi} = \hbox{\input{modules/symbols/timematchunitSym.tex}}\!
		\end{equation}

		\begin{remark}
		\noindent In the case of Hilbert spaces, the orthogonal $Z$ basis $\ket{z}_{z \in Z}$ satisfies: 
		\begin{enumerate}
			\item[(i)] $\hbox{\input{modules/symbols/timediagSym.tex}}\! \ket{z} = \ket{z} \tensor \ket{z}$ 
			\item[(ii)] $\hbox{\input{modules/symbols/timematchunitSym.tex}}\! = \sum_z \ket{z}$ 
			\item[(iii)] $\braket{z'}{z} = \delta_{zz'} d$ 
		\end{enumerate}
		
		\noindent Using points (i) and (ii) above, we obtain the following equation characterising any phase state $\ket{\psi}$:
		\begin{equation}
			\sum_{z,z'} \psi_{z'}^\star \psi_z \braket{z'}{z} \ket{z}= \sum_z\ket{z}
		\end{equation}
		
		\noindent Point (iii) allows us to conclude that phase states are exactly those in the form $\ket{\psi} = d\sum_z c_z \ket{z}$, with unimodular $c_z$ coefficients (i.e. $c_z^\star c_z = 1$) for all $z \in Z$.
		\end{remark}

		\noindent Finally, we can use phase states $\ket{\psi}$ to define \textbf{phase gates} for the $Z$ observable:
		\begin{equation}
			\phasegate{\psi} = \hbox{\input{modules/symbols/timematchSym.tex}}\! \cdot (\id{\SpaceH} \tensor \ket{\psi}) = \sum_{x,x'} \ket{x'} \psi_{(x'\ominus x)} \bra{x}
		\end{equation}

		\noindent These will again form an abelian group under composition, and the set $P$ of phase states will inherit this group structure. Using associativity of $\hbox{\input{modules/symbols/timematchSym.tex}}\!$, it is immediate to see that the group operation and unit on the phase states are given by $\hbox{\input{modules/symbols/timematchSym.tex}}\!$ and $\hbox{\input{modules/symbols/timematchunitSym.tex}}\!$ respectively: we will refer to this group as the \textbf{group of $Z$-phases}, and denote is by $(P,\oplus,0)$. Furthermore, the elements of the $X$ basis can be easily checked to be $Z$ phase states, and they also form group under $\hbox{\input{modules/symbols/timematchSym.tex}}\!$ with unit $\hbox{\input{modules/symbols/timematchunitSym.tex}}\!$: we will refer to the subgroup of $(P,\oplus,0)$ given by the elements of the $X$ basis as the \textbf{subgroup of $X$-classical points}, and denote it by $(K,\oplus,0)$.\\ 

		\noindent In order to introduce measurements, we have to move from pure states to the mixed state framework. This is a straightforward generalisation of the Hilbert space formalism, where \textbf{mixed states} in $\SpaceH$ are self-adjoint operators $\rho:\SpaceH \rightarrow \SpaceH$, possibly \textbf{positive} and possibly with unit trace: 
		\begin{enumerate}
		\item[(i)] $\rho$ is \textbf{self-adjoint}, i.e. $\rho^\dagger = \rho$
		\item[(ii)] $\rho$ is \textbf{positive}, if $\bra{\psi} \rho \ket{\psi} = b_\psi^\star b_\psi$ for all pure states $\psi$ of $\SpaceH$ and some $b_\psi \in R$ (not necessarily unique). This requirement can be omitted if positivity of mixed states is not a desideratum, e.g. in theories admitting signed probabilities.
		\item[(iii)] $\rho$ has \textbf{unit trace} (i.e. is \textbf{normalised}) if $\sum_x \rho_x^x = 1$. This requirement can be omitted if normalisation of mixed states is not a desideratum.
		\end{enumerate}

		\noindent As usual, \textbf{pure states} $\ket{\psi}$ can be identified with the 1-dimensional projectors:
		\begin{equation}
			\ket{\psi}\bra{\psi} = \sum_{x,x'} \ket{x'} \psi_{x'}^\star \psi_x \bra{x}
		\end{equation}

		\newcommand{\decoherenceSym}[1]{\operatorname{dec}_{#1}}
		\newcommand{\decoherence}[2]{\decoherenceSym{#1}[#2]}

		\noindent The \textbf{measurement/decoherence in the $X$ observable} can then be defined as the following linear transformation of mixed states, eliminating non-diagonal elements in the $X$ basis:
		\begin{equation}
		\label{eqn_decoherenceConvexCombination}
			\decoherence{X}{\rho} = \sum_{x} \ket{x} \rho_x^x \bra{x}
		\end{equation}

		\noindent Like in the Hilbert space case, measurement in the $X$ observable of a positive normalised mixed state $\rho$ always results in a convex combination\footnote{In that case, $(\rho_x^x)_{x \in X}$ is a family of positive elements which sums to 1} of eigenstates of the $X$ observable, and can thus be interpreted as a probabilistic mixture.\footnote{Where probabilities are certain positive elements of the semiring $R$, and coincide with $[0,1]$ in the case $R=\complexs$.} 

		Given a family $(\alpha_i)_{i=1}^N$ of $Z$-phases, we define the associated \textbf{Mermin measurement context} of the $N$-partite GHZ state, which we denote by $C_{(\alpha_i)_{i=1}^N}$, as follows:
		\begin{enumerate}
			\item[1.] Phase gates $\phasegate{\alpha_i}$ are applied locally to the $N$ component systems:
			\begin{equation}
			\label{eqn_MerminMeasContext1}
				\ket{\GHZ{N}} \mapsto \ket{\psi_{\alpha_1...\alpha_N}} := \left(\phasegate{\alpha_1} \tensor ... \tensor \phasegate{\alpha_N}\right) \cdot \ket{\GHZ{N}}
			\end{equation}
			\item[2.] The resulting state $\ket{\psi}$ is measured locally in the $X$ observable:
			\begin{equation}
			\label{eqn_MerminMeasContext2}
				\ket{\psi_{\alpha_1...\alpha_N}} \bra{\psi_{\alpha_1...\alpha_N}} \mapsto \left( \decoherenceSym{X} \tensor ... \tensor \decoherenceSym{X} \right) \Big[ \ket{\psi_{\alpha_1...\alpha_N}}\bra{\psi_{\alpha_1...\alpha_N}} \Big]
			\end{equation}
		\end{enumerate}

		\noindent The following Lemma \cite{CQM-StrongComplementarity} allows us us to recast the outcomes of a Mermin measurement context as the outcomes of measurement in the $X$ observable of some appropriate generalised GHZ state.\footnote{In fact, GHZ states can be further generalised from $X$-classical points to arbitrary $Z$-phases, and the result still holds.}

		\begin{lemma}
		\label{lemma_MerminMeasGeneralisedGHZ}
		Let $\ket{\alpha_1},...,\ket{\alpha_N}$ be phase states for the $Z$ observable, and suppose $x := \oplus_{i=1}^{N} \alpha_i$ is a $X$-classical point. Defining $\ket{\psi_{\alpha_1...\alpha_N}}$ as in Equation \ref{eqn_MerminMeasContext1}, one obtains the following equivalent form of the state in \ref{eqn_MerminMeasContext2}: 
		\begin{align}
			  & \left( \decoherenceSym{X} \tensor ... \tensor \decoherenceSym{X} \right) \Big[\ket{\psi_{\alpha_1...\alpha_N}}\bra{\psi_{\alpha_1...\alpha_N}}\Big] \nonumber\\ 
			= & \left( \decoherenceSym{X} \tensor ... \tensor \decoherenceSym{X} \right) \Big[ \ket{\GeneralisedGHZ{N}{x}}\bra{\GeneralisedGHZ{N}{x}} \Big] \label{eqn_MerminMeasGeneralisedGHZ}
		\end{align}
		\end{lemma}

		\noindent Equation \ref{eqn_decoherenceConvexCombination} expresses the joint outcomes $\rho_{(\alpha_i)_{i=1}^N}$ of a Mermin measurement context $C_{(\alpha_i)_{i=1}^N}$ as a mixture of $X$-classical points, and Equation \ref{eqn_GeneralisedGHZ}, together with Lemma \ref{lemma_MerminMeasGeneralisedGHZ}, can be used to explicitly compute the coefficients\footnote{Positive, since $1=1^\star 1$.} of each state in the mixture:
		\begin{equation}
			\label{eqn_MerminMeasContextOutcomesPossibilistic}
			\rho_{(\alpha_i)_{i=1}^N} = \sum_{x_1\oplus...\oplus x_N = \oplus_{i=1}^N \alpha_i} \ket{x_1}\bra{x_1} \tensor ... \tensor \ket{x_N}\bra{x_N}
		\end{equation}

		\noindent In order for $\rho_{(\alpha_i)_{i=1}^N}$ to normalisable to a positive unit trace mixed-state, which can in turn be interpreted as a probabilistic mixture of $X$-classical points, the following two requirements must hold:
		\begin{enumerate}
			\item[(a)] the size $d$ must be invertible (which we already required), and positive, i.e. $d = b^\star b$ for some $b$ (so that dividing by its inverse turns positive elements into positive elements). This is merely a technical requirement, to ensure that the coefficients in the normalised sum are positive: it can be avoided if positivity is not a desideratum (e.g. in theories admitting signed probabilities). 
			\item[(b)] the $Z$-phase $\oplus_{i=1}^N\alpha_i$ must lie in the subgroup of $X$-classical points (so that the set of $(x_1,...,x_N) \in K^N$ such that $x_1 \oplus ... \oplus x_N = \oplus_{i=1}^N\alpha_i$ is non-empty). This is a physical requirement, without which the Mermin argument measurement context will fail to be realisable.\footnote{The process is \inlineQuote{impossible} in the given theory, i.e. it doesn't return any outcomes.}
		\end{enumerate}		

		\noindent If both requirements above hold, then the Mermin measurement context will result in the following probabilistic combination of $X$-classical points (note that Equation \ref{eqn_MerminMeasContextOutcomesPossibilistic} takes the form of a possibilistic combination):
		\begin{equation}
			\label{eqn_MerminMeasContextOutcomesProbabilistic}
			\frac{1}{d^{N-1}}\rho_{(\alpha_i)_{i=1}^N} = \sum_{x_1\oplus...\oplus x_N = \oplus_{i=1}^N \alpha_i} \frac{1}{(b^\star b)^{N-1}} \ket{x_1}\bra{x_1} \tensor ... \tensor \ket{x_N}\bra{x_N}
		\end{equation}

		\begin{remark}
		\label{remark_DeterministicOutcome}
		A fundamental observation behind the Mermin argument is that the intrinsically non-deterministic outcomes (for $N \geq 2$) of any Mermin measurement context can be turned into a (interesting) deterministic outcome by applying a suitable, classical group homomorphism to them. In particular, consider the following deterministic function of $X$-classical points:
		\begin{equation}
			f = (x_1,...,x_N) \mapsto x_1 \oplus ... \oplus x_N
		\end{equation}

		\noindent Then the group homomorphism $f:K^N \rightarrow K$ applied to the probabilistic mixture $\frac{1}{d^{N-1}}\rho_{(\alpha_i)_{i=1}^N}$ of $X$-classical points yields the following deterministic $X$-classical outcome:\footnote{An analogous argument holds for the possibilistic version if we use the operation $\bigvee$ of the semiring of booleans instead of the operation $\sum$ of the semiring $R$.}
		\begin{equation}
			f(\frac{1}{d^{N-1}}\rho_{(\alpha_i)_{i=1}^N}) = \ket{\oplus_{i=1}^N \alpha_i}\bra{\oplus_{i=1}^N \alpha_i}
		\end{equation}
		\end{remark}

		\noindent Equations \ref{eqn_MerminMeasContextOutcomesPossibilistic} and \ref{eqn_MerminMeasContextOutcomesProbabilistic} show that the outcomes of a Mermin measurement context $C_{(\alpha_i)_{i=1}^N}$ are entirely characterised\footnote{Both possibilistically and probabilistically.} by the solutions $(x_1,...,x_N) \in K^N$ to the following equation:
		\begin{equation}
			x_1 \oplus ... \oplus x_N = \oplus_{i=1}^N \alpha_i
		\end{equation}
		
		\noindent In order to keep track of both the system and the $Z$-phase associated to the system in the Mermin measurement, we will adopt the following notation, generalising the one we previously used in \ref{eqn_MerminSystemZ2EquationsPMGeneralisedPhases}:
		\begin{equation}
			X_1^{\alpha_1} \oplus ... \oplus X_N^{\alpha_N} = \oplus_{i=1}^N \alpha_i
		\end{equation}

	\subsection{Operational Mermin non-locality arguments}
		\label{section_OperationalMerminNonlocalityArguments}

		\noindent Now assume that we have a $\integers$-module equation in the following form, with $a \in K$ (i.e. \textbf{valued in $K$}) and admitting some solution $y_r := \beta_r$ in the group $P$ of $Z$-phases:
		\begin{equation}
			\label{eqn_KInsolubleEquation}
			\bigoplus_{r=1}^{M} n_r y_r = a
 		\end{equation}

 		\noindent Let $k$ be the exponent of $K$, pick some $N \geq \sum_{r=1}^M n_r$ such that $N = \modclass{1}{k}$ and define:
 		\begin{equation}
 			n_0 := N - \sum_{r=1}^M n_r
 		\end{equation}

 		\noindent For  $i=1,...,N$ define $Z$-phases $\alpha_i \in P$ as follows:
 		\begin{enumerate}
 			\item[(i)] Let $\beta_0 := 0$.
 			\item[(ii)] Define a function $R:\{1,...,N\} \rightarrow \{0,...,M\}$ by:
 			\begin{equation}
 				R(i) := \text{ the least $R \geq 0$ such that } i \leq \sum_{r=0}^R n_r
 			\end{equation}
 			\item[(iii)] For $i=1,...,N$ define $\alpha_i := \beta_{R(i)}$ 
		\end{enumerate}

		\noindent Now we consider the following \textbf{Mermin measurement scenario} $\mathcal{S}$, consisting of one control and N variations:
		\begin{equation}
			\mathcal{S} = 
			\begin{cases}
				X_1^{0} \oplus X_2^{0} \oplus ... \oplus X_{N-1}^{0} \oplus X_N^{0} &= 0 \text{, the control}\\
				X_1^{\alpha_{1}} \oplus X_2^{\alpha_{2}} \oplus ... \oplus X_{N-1}^{\alpha_{N-1}} \oplus X_N^{\alpha_{N}} &= a \text{, the 1st variation}\\
				X_1^{\alpha_{2}} \oplus X_2^{\alpha_{3}} \oplus ... \oplus X_{N-1}^{\alpha_{N}} \oplus X_N^{\alpha_{1}} &= a \text{, the 2nd variation}\\
				X_1^{\alpha_{3}} \oplus X_2^{\alpha_{4}} \oplus ... \oplus X_{N-1}^{\alpha_{1}} \oplus X_N^{\alpha_{2}} &= a \text{, the 3rd variation}\\
				\hspace{2.5cm} \vdots \\
				X_1^{\alpha_{N}} \oplus X_2^{\alpha_{1}} \oplus ... \oplus X_{N-1}^{\alpha_{N-2}} \oplus X_N^{\alpha_{N-1}} &= a \text{, the Nth variation}\\
			\end{cases}
		\end{equation}

		\noindent The $N+1$ Mermin measurement contexts above can each be realised in our generalised framework: the result of applying these measurement contexts to $N+1$ distinct GHZ states can be modelled by tensor product, resulting in the following $N(N+1)$-partite mixture of $X$-classical points:
		\begin{equation}
			\label{eqn_MerminMeasScenarioOutcome}
			\rho_{\mathcal{S}} := \rho_{(0,0,...,0,0)} \tensor \rho_{(\alpha_1,\alpha_2,...,\alpha_{N-1},\alpha_N)} \tensor \rho_{(\alpha_2,\alpha_3...,\alpha_N,\alpha_1)} \tensor ... \tensor \rho_{(\alpha_N,\alpha_1...,\alpha_{N-2},\alpha_{N-1})}
		\end{equation}

		\noindent Now assume that the following deterministic function $f$ of $X$-classical points can be realised as a morphism in our generalised framework:\footnote{We already have multiplication $\hbox{\input{modules/symbols/timematchSym.tex}}\!$, but one also needs group inversion, which in categorical terms is the antipode of the strongly complementary structures. The multiplication by $n_0$ (in the abelian group/$\integers$-module $K$) can be obviated by adding up $n_0$ independent controls.}
		\begin{equation}
			f_{\mathcal{S}} = 
			\left(
			\begin{array}{ccc}
				x_1^{control}, & ...., & x_N^{control} \\
				x_1^{var_1}, & ...., & x_N^{var_1} \\
				& \vdots & \\
				x_1^{var_N}, & ...., & x_N^{var_N} 
			\end{array}
			\right)
			\mapsto \;\; \left(\left(\bigoplus\limits_{v=1}^N \bigoplus\limits_{i=1}^{N} x_i^{var_v} \right) \ominus n_0\bigoplus\limits_{i=1}^{N} x_i^{control}\right)
		\end{equation}

		\noindent By applying this $f_{\mathcal{S}}$ to the mixture $\rho_{\mathcal{S}}$ of Equation \ref{eqn_MerminMeasScenarioOutcome}, and using Remark \ref{remark_DeterministicOutcome}, we obtain a single deterministic outcome (where we used the fact that $N = \modclass{1}{k}$):
		\begin{equation}
			\label{eqn_DeterministicOutcome}
			f_{\mathcal{S}}(\rho_{\mathcal{S}}) = n_0 \cdot 0 \oplus N \cdot a = 0 \oplus a = a
		\end{equation}

		\noindent Now that we have shown how to realise a generalised scenario, we can tackle the question of locality.

		\begin{theorem}
			The mixture $\rho_{\mathcal{S}}$ admits an $X$-classical probabilistic local hidden variable model if and only if there is a solution $y_r := b_r$ in the subgroup $K$ of $X$-classical points to Equation \ref{eqn_KInsolubleEquation}.
		\end{theorem}
		\begin{proof}
			We only sketch the main points; the detailed proof can be found in \cite{StefanoGogioso-MerminNonLocality}.
			\begin{enumerate}
				\item[(i)]Suppose that $\rho_{\mathcal{S}}$ admits a $X$-classical probabilistic non-contextual hidden variable model: 
				\begin{equation}
					\rho_{\mathcal{S}} = \sum_{t=1}^T p_t \bigotimes\limits_{v=0}^N \bigotimes\limits_{i=1}^N \ket{b_{i,t}^{R_{iv}}}  \bra{b_{i,t}^{R_{iv}}} 
				\end{equation}
				where we have defined:
				\begin{enumerate}
					\item[(a)] $b_{i,t}^r \in K$ for all $t=1,...,T$, $i=1,...,N$ and $r=0,...,R$
					\item[(b)] $R_{iv} := 0$ for $v=0$ (i.e. for the control)
					\item[(c)] $R_{iv} := R(\modclass{i+v-1}{N})$ for $v=1,...,N$ (i.e. for the $N$ variations), and our modular sums are modulo $N$ with set of residues $\{1,...,N\}$ (instead of the traditional $\{0,...,N-1\}$).
				\end{enumerate}

				\noindent Because $f_{\mathcal{S}}$ is a deterministic function of $X$-classical points, and a group homomorphism $K^{N(N+1)} \rightarrow K$, Equation \ref{eqn_DeterministicOutcome} implies that, for each $t=1,...,T$, we have:
				\begin{equation}
					\bigoplus_{r=1}^{M} n_r \left(\bigoplus_{i=1}^N b_{i,t}^r\right) = b
				\end{equation}

				\noindent In particular, $(y_r := \sum_{i=1}^N b_{i,t}^r)_{r=1}^R$ is a solution in $K$ to Equation \ref{eqn_KInsolubleEquation}.

				\item[(ii)] In the other direction, assume that there is a solution $(y_r := b^r)_{r=1}^{R}$ in $K$ to Equation \ref{eqn_KInsolubleEquation}. Then, by using this solution together with Lemma \ref{lemma_MerminMeasGeneralisedGHZ}, a local hidden variable model for $\rho_{\mathcal{S}}$ can be obtained as follows:
				\begin{equation}
					\rho_{\mathcal{S}} = \sum_{x_1\oplus ... \oplus x_N = 0} \frac{1}{d^{(N-1)}} \bigotimes\limits_{v=0}^N \bigotimes\limits_{i=1}^N \ket{x_i + b^{R_{iv}}}  \bra{x_i + b^{R_{iv}}} 
				\end{equation}

			\end{enumerate}
		\end{proof}

		\noindent This method can be generalised from an individual equation in the form of \ref{eqn_KInsolubleEquation} to systems of $\integers$-module equations, constructing a Mermin measurement scenario $\mathcal{S}_{sys} = \tensor_j \mathcal{S}_{eqn_j}$ for the system by considering independent Mermin measurement scenarios $\mathcal{S}_{eqn_j}$ for each equation. This leads us to the following algebraic characterisation of Mermin non-locality. \cite{StefanoGogioso-MerminNonLocality}

		\begin{theorem}
			A process theory is Mermin non-local, i.e. it admits an operational Mermin non-locality argument, if and only if for (i) some space $\SpaceH$, (ii) some basis $X$ on $\SpaceH$, and (ii) some group structure on the $X$ basis realised by some structure $Z$, we have that the group $(P,\oplus,0)$ of $Z$-phases is an algebraically non-trivial extension of the subgroup $(K,\oplus,0)$ of $X$-classical points, i.e. that there is some system $\mathbb{S}$ of $\integers$-module equations valued in $K$ which has solutions in $P$ but not in $K$. 
		\end{theorem}

		\noindent For example, qubit stabiliser quantum mechanics is Mermin non-local, because it is possible to formulate the original Mermin non-locality argument in it: the group $\{0,\frac{\pi}{2},\pi,\frac{2\pi}{2}\} \isom \integersMod{4}$ of $Z$-phases has a solution $y := \frac{\pi}{2}$ to the equation $2 y = \pi$, which has no solution in the subgroup $\{0,\pi\} \isom \integersMod{2}$ of $X$-classical points. On the other hand, the process theory given by finite-sets and relations between them\footnote{Which is a process theory of free modules over the semiring $R = \{0,1,\vee,\wedge\}$ of the booleans.}, a model for non-deterministic classical computation, is Mermin local: all $Z$-phase groups are in the form $P = K \times H$ for some abelian group $H$, and thus any solution in $P$ to a system of equations valued in $K$ can be projected to a solution in $K$. This holds true in the more general case where the $X$ structure does not yield a basis.\footnote{In the category of sets and relations, almost all pairs $(X,Z)$ of strongly complementary structures do not yield an $X$ basis.}

	\subsection{Quantum realisability}
		\label{section_Operational_QuantumRealisability}

		Potentially, there are a lot of possible combinations $(P,K)$ of $Z$-phase groups $P$ and subgroups $K$ of $X$-classical points: it is possible to construct toy theories yielding any individual pair (but we will not do so here). However, the only features of the abelian group $P$ required by the argument are that:
		\begin{enumerate}
			\item[(i)] $P$ contains the subgroup $K$ of $X$-classical points
			\item[(ii)] $P$ contains the $Z$-phases involved in the solution to the system of equations
		\end{enumerate} 

		\noindent As a consequence, any process theory providing a phase group satisfying points (i) and (ii) above will allow for a realisation of the argument. The problem of realisability of an operational Mermin non-locality argument in some given process theory can then be formulated as follows:\\
		
		\parshape 1 .10\hsize 0.8\hsize
		\noindent \textit{Given finite abelian group $K$ and a finite consistent system $\mathbb{S}$ of $K$-valued $\integers$-module equations with no solutions in $K$, are there appropriate $X$ basis and $Z$ structure (on some system $\SpaceH$ in the given process theory), such that $K$ is isomorphic to the subgroup of $X$-classical points, and the group $P$ of $Z$-phases contains a solution to $\mathbb{S}$?}\\

		\noindent In the framework above, operational Mermin non-locality arguments can be formulated in any process theory with a suitable strong complementary pair. In particular, a large family of arguments can be formulated in the category $\fdHilbCategory$, and we shall refer to these arguments as \textbf{quantum realisable}. Let $\SpaceH$ be a $(d+1)$-dimensional Hilbert space, and $\ket{x}_{x \in X}$ be an orthonormal basis on it. Let $\mathbb{G} = (X,\oplus,0)$ be an abelian group structure on $X$ and define the $Z$ structure by:		
		\begin{align}
			\hbox{\input{modules/symbols/timematchSym.tex}}\! \cdot \left( \ket{x} \tensor \ket{x'} \right) &:= \ket{x \oplus x'} \nonumber \\
			\hbox{\input{modules/symbols/timematchunitSym.tex}}\! &:= \ket{0}
		\end{align}

		\noindent If we denote by $(\ket{j})_{j=0,...,d}$ the orthogonal basis associated with the $Z$ structure \cite{CQM-OrthogonalBases}, then the phase states for the $Z$ structure are exactly the states of $\SpaceH$ in the following form:
		\begin{equation}
			\ket{\underline{\alpha}} := \sum_{j=0}^d e^{i \alpha_j} \ket{j}
		\end{equation}

		\noindent Addition $\underline{\alpha} \oplus \underline{\beta}$ in the group of $Z$-phases is done componentwise and modulo $2 \pi$, i.e.
		\begin{equation}
		(\underline{\alpha} \oplus \underline{\beta})_j := \modclass{\alpha_j+\beta_j}{2 \pi}
		\end{equation}

		\noindent Since quantum states are identified up to global scalars, it is traditional to set $\alpha_0 = 0$. Under this identification, the abelian group of $Z$-phases forms a $d$-dimensional torus $(T^d \isom \reals^d / (2 \pi \integers)^d,\oplus,\underline{0})$, with the abelian group $\mathbb{G}$ of $X$-classical points as a subgroup of order $d+1$.


		\begin{theorem}[Quantum realisability]
			Let $\SpaceH$ be a $(d+1)$-dimensional Hilbert space, with $d \geq 0$, and $\ket{x}_{x \in X}$ any orthonormal basis on it. We will refer to the associated classical structure as the \textbf{$X$ structure}, and to the elements of the basis as \textbf{$X$-classical points}. Let $\mathbb{G} = (X,\oplus,0)$ be a finite abelian group of order $d+1$, and consider some consistent system $\mathbb{S}$ of $\mathbb{G}$-valued $\integers$-module equations with no solution in $\mathbb{G}$:
			\begin{equation}\label{eqn_systemAlgExt}
				\begin{cases}
				\bigoplus_{r=1}^{M} n^1_r \, y_r = a^1 \\
				\hspace{1.2cm} \vdots  \\
				\bigoplus_{r=1}^{M} n^S_r \, y_r = a^S 
				\end{cases}
			\end{equation} 
			Then there exists a quasi-special commutative $\dagger$-Frobenius algebra, which we will refer to as the \textbf{$Z$ structure}, such that the following is true: 
			\begin{enumerate}
				\item[(i)] the $X$ and $Z$ structures are strongly complementary, so that the $X$ basis forms a subgroup of the abelian group $P$ of $Z$ phases.
				\item[(ii)] the subgroup $K$ of $X$-classical points is isomorphic to $\mathbb{G}$.
				\item[(iii)] the consistent system $\mathbb{S}$ admits a solution $(y_r := \beta_r)_{r=1}^M$ in the group $P$ of $Z$ phases.
			\end{enumerate}
		\end{theorem}
		\begin{proof}
			Take the $Z$ structure to be the unique quasi-special commutative $\dagger$-Frobenius algebra such that points (i) and (ii) above hold \cite{StefanoGogioso-FourierTransform}\cite{CQM-KissingerPhdthesis}. We are now looking for $Z$-phases $\underline{\beta}^{(1)},...,\underline{\beta}^{(M)}$ such that 
			\begin{equation}\label{eqn_systemAlgExtTd}
				\begin{cases}
				\bigoplus_{r=1}^{M} n^1_r \, \underline{\beta}^{(r)} = \underline{a}^{(1)} \\
				\hspace{1.2cm} \vdots  \\
				\bigoplus_{r=1}^{M} n^S_r \, \underline{\beta}^{(r)} = \underline{a}^{(S)}
				\end{cases}
			\end{equation} 
			where we have adapted notation to accommodate the fact that both the $Z$-phases $\underline{\beta}^{(r)}$ and the $X$-classical points $\underline{a}^{(s)}$ are points of the torus $T^d \isom \reals^d / (2 \pi \integers)^d$, which can be written as $d$-dimensional vectors with coordinates in $S^1 \isom \reals/(2 \pi \integers)$.

			Observe that solving $\mathbb{S}$ in $T^d \isom \reals^d / (2 \pi \integers)^d$ is equivalent to solving the following $d$ independent systems, one for each $j=1,...,d$, in $S^1 \isom \reals/(2 \pi \integers)$:
			\begin{equation}\label{eqn_systemAlgExtS1}
				\begin{cases}
				\oplus_{r=1}^{M} n^1_r \, \beta^{(r)}_j = a^{(1)}_j \\
				\hspace{1.2cm}\vdots \\
				\oplus_{r=1}^{M} n^S_r \, \beta^{(r)}_j = a^{(S)}_j 
				\end{cases}
			\end{equation} 
			We apply Gaussian elimination (see section \ref{section_AppendixGaussianEliminationS1} in the Appendix) to solve each of the $d$ systems in $S^1$ (any solution will do), and obtain our Z-phases as the corresponding points $\underline{\beta}^{(1)},...,\underline{\beta}^{(M)}$ of $T^d$.
		\end{proof}

		\begin{corollary}
			All operational Mermin non-locality arguments are quantum realisable.
		\end{corollary}

\section{All-vs-Nothing arguments}
	\label{section_AllvsNothingArguments}
	
	\noindent The second part of this paper presents the work of \cite{NLC-AvNSheaf} on All-vs-Nothing arguments, a different generalisation of the original Mermin non-locality argument, and clarifies its relationship with the operational Mermin non-locality presented in the first part. The sheaf-theoretic contextuality framework is reviewed, and Mermin's original argument is rephrased within it. The framework of All-vs-Nothing arguments is then introduced. Finally, operational Mermin non-locality arguments are shown to provide an infinite non-collapsing hierarchy (over the finite rings $\integersMod{n}$) of quantum realisable All-vs-Nothing arguments.

	\subsection{Sheaf-theoretic contextuality}

		In this section, we summarise the basic framework of sheaf-theoretic contextuality by \cite{NLC-SheafSeminal}. We begin by considering a finite set $\mathcal{X}$ of \textbf{measurements}; in the abstract framework this is just a set, but from the point of view of realisable non-locality scenarios these are measurements of some state $\ket{\psi}$ in some space $\SpaceH$. For example, if $\ket{\psi} = \ket{\GHZ{3}}$ is a 3-qubit GHZ state in the single-qubit Pauli $Z$ basis, then the measurements involved in the original Mermin argument form the following six element set:
		\begin{equation}
			\mathcal{X} = \bigsqcup_{j=1}^3 \{X_j,Y_j\}
		\end{equation}
		where $X_j/Y_j$ are measurements in the single-qubit Pauli $X/Y$ observables on the $j$-th qubit. More in general, for measurements of $N$-partite states, where party $j$ has access to a finite measurement set $M_j$ and all parties choose their measurements independently, one obtains the disjoint union $\mathcal{X} = \sqcup_{j=1}^N M_j$. The disjoint union preserves information about which party each measurement is associated to, so we will adopt the notation $m_j$ for generic elements of $\mathcal{X}$, where $m$ is the measurement and $j$ is the party.

		Each measurement $m_j \in M_j$ comes with a set $O_j^m$ of \textbf{outcomes}: if $U \subseteq \mathcal{X}$ is a subset of measurements, then the family of all potential \textbf{joint outcomes} takes the form:
		\begin{equation}
			\sheafOfEvents{U} := \prod_{m_j \in U} O_j^m
		\end{equation} 

		\noindent A fundamental feature of quantum mechanics is that not all measurements are compatible, so we shouldn't expect joint outcomes to play a role for sets $U$ containing incompatible measurements. From an abstract point of view, this is captured in the sheaf-theoretic framework by specifying a set $\mathcal{M}$ of \textbf{measurement contexts}, sets $C \subseteq \mathcal{X}$ of measurements which are mutually compatible (and therefore have a well-defined notion of joint outcome). One need not specify \textit{all} sets of mutually compatible outcomes as measurement contexts, but only those which are needed by a specific non-locality argument; for example, the measurement contexts involved in Mermin's original non-locality argument are:
		\begin{align}
			C_{control} & := \{X_1,X_2,X_3\} \nonumber \\
			C_{var_1} & := \{Y_1,Y_2,X_3\} \nonumber \\
			C_{var_2} & := \{Y_1,X_2,Y_3\} \nonumber \\
			C_{var_3} & := \{X_1,Y_2,Y_3\}  
		\end{align}

		\noindent The set $\Powerset{\mathcal{X}}$ of all possible subsets $U$ of the finite set $\mathcal{X}$ is a poset (and therefore a poset category) under inclusion $V \subseteq U$ of subsets. We can define a functor $\sheafOfEventsSym: \OpCategory{\Powerset{\mathcal{X}}} \rightarrow \SetCategory$, i.e. a \textbf{presheaf}, by setting:
		\begin{enumerate}
			\item[(i)] if $U \in  \Powerset{\mathcal{X}} $, then we define $ \sheafOfEvents{U} := \prod_{m_j \in U} O_j^m$ as above
			\item[(ii)] if $V \subseteq U $, then we define $\sheafOfEvents{V \subseteq U}:=\restrictionMap{U}{V}$ to be the following \textbf{restriction map} $U \stackrel{\SetCategory}{\longrightarrow} V$:
			\begin{equation}
				\restrictionMap{U}{V} = s \mapsto \restrict{s}{V}
			\end{equation}
			which sends a \textbf{section $s$ over $U$} (or \textbf{$U$-section}):
			\begin{equation}
				s = \suchthat{(m_j,s(m_j))}{m_j \in U} \in  \prod_{m_j \in U} O_j^m
			\end{equation}
			to its restriction $\restrict{s}{V}$ to a section over $V$:
			\begin{equation}
				\restrict{s}{V} = \suchthat{(m_j,s(m_j))}{m_j \in V} \in  \prod_{m_j \in V} O_j^m
			\end{equation}
		\end{enumerate}

		\noindent Furthermore, $\Powerset{\mathcal{X}}$ is a locale, the locale of open sets for $\mathcal{X}$ endowed with the discrete topology. The locale structure defines \textbf{local covers} for any $U \in \Powerset{\mathcal{X}}$ as the families $(U_i)_{i \in I}$ such that $\cup_{i \in I} U_i = U$. A \textbf{global cover} for $\Powerset{\mathcal{X}}$ is a local cover of $\mathcal{X}$, and from now on we will require the set $\mathcal{M}$ of measurement contexts to be a global cover of $\mathcal{X}$, i.e. $\cup_{C \in \mathcal{M}} C = \mathcal{X}$. 

		Since $\Powerset{\mathcal{X}}$ is a locale, one can define a notion of \textbf{sheaf} on it. Let $F: \OpCategory{\Powerset{\mathcal{X}}} \rightarrow \SetCategory$ be any presheaf. If $(U_i)_{i \in I}$ is a local cover of some $U \in \Powerset{\mathcal{X}}$, then a \textbf{compatible family} of elements of $F$ indexed by $(U_i)_{i \in I}$ is a family $(s_i \in F[U_i])_{i \in I}$ such that: 
		\begin{equation}
			F[U_i \cap U_j \subseteq U_i](s_i) = F[U_i \cap U_j \subseteq U_j](s_j) \text{ for all } i,j \in I
		\end{equation}

		\noindent In particular, if $s \in F[U]$ then letting $s_i := F[U_i \subseteq U](s)$ for all $i \in I$ defines a compatible family. The presheaf $F$ is then a \textbf{sheaf} on the locale $\Powerset{\mathcal{X}}$ if it satisfies the following \textbf{gluing condition}:\footnote{In mathematical literature, this is often stated as two separate conditions: one of existence, the gluing condition, and one of uniqueness, the locality condition.} every compatible family $(s_i \in F[U_i])_{i \in I}$ admits a unique \textbf{gluing} $s \in F[U]$ such that $s_i = F[U_i \subseteq U](s)$ for all $i \in I$.

		Because it is defined in terms of sections, the presheaf $\sheafOfEventsSym$ is in fact a sheaf on the locale $\Powerset{\mathcal{X}}$, and we shall refer to it as the \textbf{sheaf of events}. Indeed, consider a family $(U_i)_{i \in I}$ of subsets of $\mathcal{X}$ (a local cover of $U := \cup_{i \in I} U_i$), and a compatible family $(s_i \in \sheafOfEvents{U_i})_{i \in I}$ of sections:
		\begin{equation}
			\restrict{s_i}{U_i \cap U_j} = \restrict{s_j}{U_i \cap U_j} \text{ for all } i,j \in I
		\end{equation}
		i.e. we have that $s_i(m_i) = s_j(m'_j)$ whenever $m_i = m'_j$ for some $m_i \in U_i$ and $m'_j \in U_j$. Then there exists a unique gluing $s \in \sheafOfEvents{U}$ such that for all $i \in I$ we have $\restrict{s}{U_i} = s_i$:
		\begin{equation}
			s := \suchthat{(m_j,s_j(m_j))}{m_j \in U}
		\end{equation}

		\noindent We are now in possession of all the ingredients of a \textbf{measurement scenario}, encoded by the pair $(\sheafOfEventsSym,\mathcal{M})$, and it's time to introduce probabilities. In quantum mechanics, \textbf{probabilities} take values in the commutative semiring $R = (\reals^+,+,0,\cdot,1)$ of the non-negative reals, and in fact they traditionally fall within the interval $[0,1]$, a consequence in $R$ of the requirement that probabilities add up to $1$. In other circumstances, one may be interested in the \textbf{possibilities} associated with events, living in the commutative semiring $\mathbb{B} = (\{0,1\},\vee,0,\wedge,1)$ of the booleans. In the sheaf-theoretic treatment of contextuality, one works with an arbitrary commutative semiring $R = (|R|,+,0,\cdot,1)$. Given a set $U$, an \textbf{$R$-distribution} on $U$ is a function $d: U \rightarrow R$ which has finite \textbf{support} $\support{d} := \suchthat{s \in U}{d(s) \neq 0}$ and such that:
		\begin{equation}
			\sum_{s \in \support{d}} d(s) = 1
		\end{equation}  

		\noindent One can define a functor $\distributionFunctorSym{R}: \SetCategory \rightarrow \SetCategory$ by setting:
		\begin{enumerate}
			\item[(i)] for any set $U$, define $\distributionFunctor{R}{U}$ to be the set of $R$-distributions of $U$
			\item[(ii)] for any function $f: U \rightarrow V$, define $\distributionFunctor{R}{f}: \distributionFunctor{R}{U} \rightarrow \distributionFunctor{R}{V}$ to be the following function:
			\begin{equation}
				\distributionFunctor{R}{f} = d \mapsto \left[ t \mapsto \sum_{f(s) = t} d(s) \right]
			\end{equation}
		\end{enumerate}

		\noindent Composing this functor with the sheaf of events we obtain the \textbf{presheaf of distributions} $\presheafOfDistributionsSym{R}:\OpCategory{\Powerset{\mathcal{X}}} \rightarrow \SetCategory$, sending each set $U$ of measurements to the set $\presheafOfDistributions{R}{U}$ of \textbf{distributions on $U$-sections}, and acting as marginalisation of distributions on the inclusions $V \subseteq U$:
		\begin{equation}
			\presheafOfDistributions{R}{V \subseteq U} = d \mapsto \restrict{d}{V} := \left[ t \mapsto \sum_{\restrict{s}{V} = t} d(s)\right]
		\end{equation}

		\noindent We will refer to $\restrict{d}{V}$ as the \textbf{marginal} of $d$ since, when the notation is expanded a bit, it takes the following, more familiar form:
		\begin{equation}
			t \in \sheafOfEvents{V} \implies \restrict{d}{V}(t) := \hspace{-0.75cm}\sum_{s \in \sheafOfEvents{V} \text{ s.t. } \restrict{s}{V} = t} \hspace{-0.75cm} d(s)
		\end{equation}

		\noindent In quantum mechanics, if $C$ is a set of compatible measurements on some state $\ket{\psi}$, then there is a probability distribution $d \in \presheafOfDistributions{\reals^+}{C}$ on the joint outcomes of the measurements, and the typical contextuality argument involves showing that the probability distributions on different contexts cannot be obtained, in a no-signalling scenario, as marginals of some non-contextual hidden variable. In the sheaf-theoretic framework, a \textbf{(no-signalling) empirical model} is defined to be a compatible family of distributions $(e_C)_{C \in \mathcal{M}}$ for the global cover $\mathcal{M}$ of measurement contexts; the usual no-signalling property is shown in \cite{NLC-SheafSeminal} to be a special case of the compatibility condition. A \textbf{global section} for an empirical model\footnote{From now on, no-signalling is implicitly assumed.} $(e_C)_{C \in \mathcal{M}}$ is a distribution $d \in \presheafOfDistributions{R}{\mathcal{X}}$ over the joint outcomes of all measurements which marginalises to the distributions specified by the empirical model:
		\begin{equation}
			\restrict{d}{C}=e_C \text{ for all } C \in \mathcal{M}
		\end{equation} 

		\noindent The fundamental observation behind the sheaf-theoretic framework is that the existence of a global section for an empirical model is equivalent to the existence of a \textbf{non-contextual hidden variable model}: we will say that an empirical model $(e_C)_{C \in \mathcal{M}}$ is \textbf{contextual} if it doesn't admit a global section. Contextuality of probabilistic models is interesting in itself, but more refined notions can be obtained by relating $\reals^+$ to two other semirings: the reals, modelling signed probabilities, and the booleans, modelling possibilities. Observe that the construction $\distributionFunctorSym{R}$ is functorial in $R$: if $r: R \rightarrow R'$ is a morphism of semirings, then for any fixed set $U$ we can define:
		\begin{equation}
			\distributionFunctor{r}{U} = \left[d:U \rightarrow R\right] \mapsto \left[r \circ d: U \rightarrow R'\right]
		\end{equation}

		\noindent In particular, there is an injective morphism of semirings $i^+: \reals^+ \inject R$ sending $x \in \reals^+$ to $+x \in \reals$ and a surjective morphism of semirings $p: \reals^+ \rightarrow \mathbb{B}$ sending $0 \mapsto 0$ and $x \neq 0 \mapsto 1$. If $(e_C)_{C \in \mathcal{M}}$ is a probabilistic empirical model, i.e. one in the semiring $\reals^+$, then $(e_C)_{C \in \mathcal{M}}$ can be seen as an empirical model $(i^+ \circ e_C)_{C \in \mathcal{M}}$ in the semiring $\reals$: regardless of whether $(e_C)_{C \in \mathcal{M}}$ was contextual or not over $\reals^+$, it can be shown \cite{NLC-SheafSeminal} that over the reals it always admits a global section. 

		On the other hand, any probabilistic empirical model $(e_C)_{C \in \mathcal{M}}$ can be assigned a corresponding possibilistic empirical model $(p\circ e_C)_{C \in \mathcal{M}}$ in the semiring $\mathbb{B}$ of the booleans (and each boolean function $p \circ e_C$ can equivalently be seen as the characteristic function of the subset $\support{e_C} \subseteq \sheafOfEvents{C}$). Note that contextuality is a contravariant property with respect to change of semiring: if $(e_C)_{C \in \mathcal{M}}$ is an empirical model in a semiring $R$ and $r: R \rightarrow R'$ is a morphism of semiring, then contextuality of $(r \circ e_C)_{C \in \mathcal{M}}$ implies contextuality of $(e_C)_{C \in \mathcal{M}}$ (because a global section $d$ of the latter is mapped to a global section $r \circ d$ of the former). We will say that a probabilistic empirical model $(e_C)_{C \in \mathcal{M}}$ is \textbf{probabilistically non-extendable} if it is contextual, and \textbf{possibilistically non-extendable} if the corresponding possibilistic model $(p \circ e_C)_{C \in \mathcal{M}}$ is contextual: because of contravariance, possibilistic non-extandability implies probabilistic non-extendability. Furthermore, the opposite is not true: the Bell model given in \cite{NLC-SheafSeminal} is probabilistically non-extendable but not possibilistically non-extendable.

		Seeing distributions $d \in \presheafOfDistributions{\mathbb{B}}{U}$ as indicator functions of the subsets $\support{d} \subseteq \sheafOfEvents{U}$ endows them with a partial order:
		\begin{equation}
			\label{eqn_SCimpliesC}
			d' \preceq d \text{ if and only if } \support{d'} \subseteq \support{d}
		\end{equation}

		\noindent The existence of a global section $d \in \presheafOfDistributions{\mathbb{B}}{U}$ for a possibilistic empirical model $(e_C)_{C \in \mathcal{M}}$ implies that:
		\begin{equation}
			\label{eqn_StrongContextualityCondition}
			\restrict{d}{C} \preceq e_C \text{ for all } C \in \mathcal{M}
		\end{equation}

		\noindent We say that a possibilistic empirical model $(e_C)_{C \in \mathcal{M}}$  is \textbf{strongly contextual} iff there is no distribution $d \in \presheafOfDistributions{\mathbb{B}}{\mathcal{X}}$ such that Equation \ref{eqn_StrongContextualityCondition} holds. In particular, the GHZ model given in \cite{NLC-SheafSeminal}, corresponding to Mermin's original non-locality argument, is strongly contextual. Because of Equation \ref{eqn_SCimpliesC}, strong contextuality implies contextuality, but the opposite is not true: the possibilistic Hardy model give in \cite{NLC-SheafSeminal} is contextual, but not strongly contextual. We will say that a probabilistic empirical model is strongly contextual iff the associated possibilistic empirical model is, and this defines the following (strict) hierarchy of notions of contextuality for probabilistic empirical models:
		\begin{equation}
			\text{probabilistically non-extendable } \Leftarrow \text{ possibilistically non-extendable } \Leftarrow \text{ strongly contextual}
		\end{equation} 

		\noindent All $d \in \presheafOfDistributions{\mathbb{B}}{\mathcal{X}}$ satisfying Equation \ref{eqn_StrongContextualityCondition} form a (possibly empty) lattice, and thus a probabilistic empirical model is strongly contextual iff the following set is empty:
		\begin{equation}
			\supportSubpresheaf{\mathcal{X}} := \suchthat{s \in \sheafOfEvents{\mathcal{X}}}{ \restrict{s}{C} \in \support{e_C} \text{ for all } C \in \mathcal{M}}
		\end{equation}

		\noindent For a possibilistic (no-signalling) empirical model $(e_C)_{C \in \mathcal{M}}$, we can define \cite{NLC-AvNSheaf} a \textbf{support subpresheaf} $\supportSubpresheafSym \subseteq \sheafOfEventsSym$ by setting:
		\begin{equation}
			\supportSubpresheaf{U} := \suchthat{s \in \sheafOfEvents{U}}{ \restrict{s}{C \cap U} \in \support{\restrict{e_C}{U \cap C}} \text{ for all } C \in \mathcal{M}}
		\end{equation}
		Then a possibilistic empirical model is strongly contextual if and only if $\supportSubpresheaf{\mathcal{X}} = \emptyset$. The support subpresheaf satisfies \cite{NLC-AvNSheaf} the following properties:
		\begin{enumerate}
			\item[(i)] $\supportSubpresheaf{C} \neq \emptyset$ for all $C \in \mathcal{M}$
			\item[(ii)] $\supportSubpresheafSym$ is \textbf{flasque beneath the cover}, i.e. $\supportSubpresheaf{V \subseteq U}$ is a surjective function whenever $V \subseteq U \subseteq C$ for some $C \in \mathcal{M}$
			\item[(iii)] $\supportSubpresheafSym$ is a \textbf{sheaf above the cover}, i.e. if $(s_C \in \supportSubpresheaf{C})_{C \in \mathcal{M}}$ is a compatible family for the subpresheaf $\supportSubpresheafSym$, then there is a unique global section $s \in \supportSubpresheaf{\mathcal{X}}$ such that $\restrict{s}{C} = s_C$ for all $C \in \mathcal{M}$ 
		\end{enumerate}

		\noindent The possibilistic empirical model $(e_C)_{C \in \mathcal{M}}$ can be reconstructed from the subpresheaf $\supportSubpresheafSym$ as follows:
		\begin{equation}
			\label{eqn_SubpresheafEM}
			\support{e_C} = \supportSubpresheaf{C}
		\end{equation} 
		and furthermore it can be shown \cite{NLC-AvNSheaf} that any subpresheaf $\supportSubpresheafSym \subseteq \sheafOfEventsSym$ satisfying properties (i)-(iii) above defines a possibilistic no-signalling empirical model via Equation \ref{eqn_SubpresheafEM}.

	\subsection{Mermin's original non-locality argument}

		In this section, we will re-cast Mermin's original non-locality argument into a probabilistic empirical model. We then examine the associated possibilistic empirical model and give an explicit, group-theoretic proof of strong contextuality (other, more general proofs already appeared in \cite{NLC-SheafSeminal} and \cite{NLC-AvNSheaf}, but it is instructive to give an explicit one for the original case). 

		As mentioned in the previous section, the set of measurements for Mermin's non-locality argument is given as follows:
		\begin{equation}
			\mathcal{X} = \bigsqcup_{j=1}^3 \{X_j,Y_j\}
		\end{equation}
		and the global cover $\mathcal{M}$ is given by the following measurement contexts:\begin{align}
			C_{control} & := \{X_1,X_2,X_3\} \nonumber \\
			C_{var_1} & := \{Y_1,Y_2,X_3\} \nonumber \\
			C_{var_2} & := \{Y_1,X_2,Y_3\} \nonumber \\
			C_{var_3} & := \{X_1,Y_2,Y_3\}  
		\end{align}

		\noindent Because the measurement $Y_j$ can be equivalently be obtained by applying a phase gate and then measuring in $X_j$, we will assume all outcomes to be valued in the group $\integersMod{2}$, and define the sheaf of events as:
		\begin{equation} 
			\sheafOfEvents{U} := (\integersMod{2})^U
		\end{equation}

		\noindent In particular, for each measurement context $C$ we have $\sheafOfEvents{C} \isom (\integersMod{2})^3$. The probabilistic empirical model $(e_C)_{C \in \mathcal{M}}$ arising from Mermin's original non-locality argument can then be written as follows:
		\begin{equation}
			\begin{array}{ccc|c|c}
			&&					& (b_1,b_2,b_3) \text{ s.t. } b_1\oplus b_2 \oplus b_3 =_{\integersMod{2}} 0 & (b_1,b_2,b_3) \text{ s.t. } b_1\oplus b_2 \oplus b_3 =_{\integersMod{2}} 1\\
			\hline 
			X_1 & X_2 & X_3 	&	1/4		&	0 \\  
			Y_1 & Y_2 & X_3 	&	0		&	1/4 \\  
			Y_1 & X_2 & Y_3 	&	0		&	1/4 \\  
			X_1 & Y_2 & Y_3 	&	0		&	1/4 \\ 
			\end{array}
		\end{equation}

		\noindent The associated possibilistic empirical model $(p \circ e_C)_{C \in \mathcal{M}}$ can be written as follows:		
		\begin{equation}
			\begin{array}{ccc|c|c}
			&&					& (b_1,b_2,b_3) \text{ s.t. } b_1\oplus b_2 \oplus b_3 =_{\integersMod{2}} 0 & (b_1,b_2,b_3) \text{ s.t. } b_1\oplus b_2 \oplus b_3 =_{\integersMod{2}} 1\\
			\hline 
			X_1 & X_2 & X_3 	&	1		&	0 \\  
			Y_1 & Y_2 & X_3 	&	0		&	1 \\  
			Y_1 & X_2 & Y_3 	&	0		&	1 \\  
			X_1 & Y_2 & Y_3 	&	0		&	1 \\ 
			\end{array}
		\end{equation}

		\noindent The subpresheaf $\supportSubpresheafSym \subseteq \sheafOfEventsSym$ associated to $(e_C)_{C \in \mathcal{M}}$ takes the following values on the cover:
		\begin{align}
			\supportSubpresheaf{C_{control}}  &= \suchthat{(b_1^X,b_2^X,b_3^X)}{b_1^X \oplus b_2^X \oplus b_3^X =_{\integersMod{2}} 0} \nonumber\\
			\supportSubpresheaf{C_{var_1}}  &= \suchthat{(b_1^Y,b_2^Y,b_3^X)}{b_1^Y \oplus b_2^Y \oplus b_3^X =_{\integersMod{2}} 1} \nonumber\\
			\supportSubpresheaf{C_{var_2}}  &= \suchthat{(b_1^Y,b_2^X,b_3^Y)}{b_1^Y \oplus b_2^X \oplus b_3^Y =_{\integersMod{2}} 1} \nonumber\\
			\supportSubpresheaf{C_{var_3}}  &= \suchthat{(b_1^X,b_2^Y,b_3^Y)}{b_1^X \oplus b_2^Y \oplus b_3^Y =_{\integersMod{2}} 1} 
		\end{align} 
		where we have labelled the elements as $b_j^X$ or $b_j^Y$ to denote that they are outcomes of different measurements $X_j$ or $Y_j$ (and thus live in different, albeit isomorphic, sets). Strong contextuality is equivalent to showing that $\supportSubpresheaf{\mathcal{X}} = \emptyset$: the existence of any element 
		\begin{equation}
			(b_1^X,b_1^Y,b_2^X,b_2^Y,b_3^X,b_3^Y) \in \supportSubpresheaf{\mathcal{X}}
		\end{equation}
		would provide a solution to the following (inconsistent) system of equations: 
		\begin{equation}
		\label{eqn_MerminOriginalArgumentSystemAvN}
		\begin{cases}
			b_1^X \oplus b_2^X \oplus b_3^X =_{\integersMod{2}} 0\\
			b_1^Y \oplus b_2^Y \oplus b_3^X =_{\integersMod{2}} 1\\
			b_1^Y \oplus b_2^X \oplus b_3^Y =_{\integersMod{2}} 1\\
			b_1^X \oplus b_2^Y \oplus b_3^Y =_{\integersMod{2}} 1 
		\end{cases}
		\end{equation} 

		\noindent Therefore the empirical model arising from Mermin's non-locality argument is strongly contextual.

	\subsection{All-vs-Nothing arguments}

		The fundamental observation behind the generalised All-vs-Nothing arguments of \cite{NLC-AvNSheaf} is that strong contextuality of Mermin's original non-locality argument follows straightforwardly from the existence of the system of $\integersMod{2}$ equations \ref{eqn_MerminOriginalArgumentSystemAvN} which has no global solution (corresponding to $\supportSubpresheaf{\mathcal{X}} = \emptyset$), but where each equation admits a solution (corresponding to $\supportSubpresheaf{C} \neq \emptyset$ for all measurement contexts $C$). In this section we summarise the basic framework of All-vs-Nothing arguments from \cite{NLC-AvNSheaf}, taking the liberty of slightly generalising the definitions therein from rings to modules over rings.

		\newcommand{\eqnIndex}[1]{\operatorname{index}(#1)}

		Let $R$ be some ring. Note that, in this section, $R$ will no more denote the semiring over which distributions are taken (which is fixed to $\mathbb{B}$), but will be some commutative ring with unit; we will denote by $+$ the addition in the ring $R$, and by $\oplus$ the addition in $R$-modules. If $G$ is some $R$-module, we will define an \textbf{$R$-linear equation valued in $G$} to be a triple $\phi = (C,n,b)$ where:
		\begin{enumerate}
			\item[(i)] $C$ is some finite set, and we define $\eqnIndex{\phi} := C$
			\item[(ii)] $n: C \rightarrow R$ is any function
			\item[(iii)] $b \in G$
		\end{enumerate}

		\newcommand{\RlinearTheory}[2]{\mathbb{T}_{#1}(#2)}

		\noindent If $\phi = (C,n,b)$ is an $R$-linear equation valued in $G$, we will say that a function $s: C \rightarrow G$ (henceforth an \textbf{assignment}) \textbf{satisfies} $\phi$, written $s \models \phi$, if and only if the following equation holds in $G$:
		\begin{equation}
			\bigoplus_{m \in C} n_m s_m = b
		\end{equation}
		where we denoted $n_m := n(m)$ and $s_m := s(m)$. Any set $W$ of assignments $C \rightarrow G$ can be associated its corresponding set $\RlinearTheory{R}{W}$ of satisfied equations, which is itself an $R$-module:\footnote{This gives rise to some interesting results on affine closures, see \cite{NLC-AvNSheaf}.}
		\begin{equation}
			\label{eqn_RlinearTheory}
			\RlinearTheory{R}{W} := \suchthat{\phi}{s \models \phi \text{ for all }s \in W}
		\end{equation} 

		\noindent Let $(e_C)_{C \in \mathcal{M}}$ be a possibilistic empirical model for a measurement scenario $(\sheafOfEventsSym,\mathcal{M})$, such that all measurements have the same $R$-module $G$ as their set of outcomes (for example we had $G = \integersMod{2}$, a $\integers$-module, for Mermin's original non-locality argument). Let $\supportSubpresheafSym \subseteq \sheafOfEventsSym$ be the support subpresheaf for the empirical model and define its \textbf{$R$-linear theory} to be:
		\begin{equation}
			\RlinearTheory{R}{\supportSubpresheafSym} := \bigcup_{C \in \mathcal{M}}  \RlinearTheory{R}{\supportSubpresheaf{C}}
		\end{equation} 
		
		\newcommand{\AvN}[2]{\operatorname{AvN}_{#1,#2}}
		\newcommand{\AvNring}[1]{\operatorname{AvN}_{#1}}

		\noindent We say that the empirical model is \textbf{All-vs-Nothing} with respect to ring $R$ and $R$-module $G$, written $\AvN{R}{G}$, iff the $R$-linear theory admits no solution in $G$, i.e. iff there exists no global assignment $s: \mathcal{X} \rightarrow G$ such that:
		\begin{equation}
			\label{eqn_AvNDefinition}
			\restrict{s}{C} \models \phi \text{ for all } C \in \mathcal{M} \text{ and all } \phi \in \RlinearTheory{R}{\supportSubpresheaf{C}}
		\end{equation}
		
		\noindent To connect back with the notation in \cite{NLC-AvNSheaf}, we will simply write $\AvNring{R}$ for $\AvN{R}{R}$. It is now straightforward to prove \cite{NLC-AvNSheaf} that any model which is $\AvN{R}{G}$ for some ring $R$ and some $R$-module $G$ is strongly contextual: if the model isn't strongly contextual, then there is some global assignment $s \in \supportSubpresheaf{\mathcal{X}}$, and this implies $\restrict{s}{C} \in \supportSubpresheaf{C}$ for all $C \in \mathcal{M}$, which in turn implies the desired Equation \ref{eqn_AvNDefinition} (by appealing to Equation \ref{eqn_RlinearTheory}). Because every probabilistic empirical model is strongly contextual if and only if it is maximally contextual \cite{NLC-SheafSeminal}, i.e. if and only if it lies on a face of the no-signalling polytope, then being $\AvN{R}{G}$ is a particularly neat way of proving that a quantum realisable model is maximally contextual, a highly desired property in the field of quantum security.

	\subsection{Operational Mermin non-locality arguments}

		Now we turn our attention back to the operational Mermin non-locality presented in the first part of this paper. In the setting we presented, scalars have a semiring structure $R$ on them, and therefore they are prima facie compatible with the sheaf-theoretic framework. Consider a generic Mermin measurement scenario, consisting of one control and N variations:
		\begin{equation}
			\label{eqn_MerminMeasurementScenario}
			\begin{cases}
				X_1^{0} \oplus X_2^{0} \oplus ... \oplus X_{N-1}^{0} \oplus X_N^{0} &= 0 \text{, the control}\\
				X_1^{\alpha_{1}} \oplus X_2^{\alpha_{2}} \oplus ... \oplus X_{N-1}^{\alpha_{N-1}} \oplus X_N^{\alpha_{N}} &= a \text{, the 1st variation}\\
				X_1^{\alpha_{2}} \oplus X_2^{\alpha_{3}} \oplus ... \oplus X_{N-1}^{\alpha_{N}} \oplus X_N^{\alpha_{1}} &= a \text{, the 2nd variation}\\
				X_1^{\alpha_{3}} \oplus X_2^{\alpha_{4}} \oplus ... \oplus X_{N-1}^{\alpha_{1}} \oplus X_N^{\alpha_{2}} &= a \text{, the 3rd variation}\\
				\hspace{2.5cm} \vdots \\
				X_1^{\alpha_{N}} \oplus X_2^{\alpha_{1}} \oplus ... \oplus X_{N-1}^{\alpha_{N-2}} \oplus X_N^{\alpha_{N-1}} &= a \text{, the Nth variation}\\
			\end{cases}
		\end{equation}
		where $\alpha_i := a_{R(i)}$ are $Z$-phases as in Section \ref{section_OperationalMerminNonlocalityArguments} and $(y_r := a_r)_{r=1}^M$ is a solution in the group $(P,\oplus,0)$ of $Z$-phases to the following $\integers$-module equation (valued in the subgroup $(K,\oplus,0)$ of $X$-classical points):
		\begin{equation}
			\label{eqn_UnsolvableEquationAvNSection}
			\bigoplus\limits_{r=1}^M n_r y_r = a \in K
		\end{equation}

		\noindent We assume w.l.o.g. that all $a_r$ are distinct and non-zero, and that all $n_r$ are non-zero. The set of measurements for this scenario can be written as follows:
		\begin{equation}
			\mathcal{X} = \bigsqcup_{i=1}^N\{X_i^{0},X_i^{a_1},...,X_i^{a_M}\}
		\end{equation}
		and the measurement contexts can be written as follows:
		\begin{align}
			C_{control} &= \suchthat{X_i^{0}}{i=1,...,N} \nonumber \\
			C_{var_1} &= \suchthat{X_i^{\alpha_{i}}}{i=1,...,N} \nonumber \\
			C_{var_2} &= \suchthat{X_i^{\alpha_{i+1}}}{i=1,...,N} \nonumber \\
			C_{var_3} &= \suchthat{X_i^{\alpha_{i+2}}}{i=1,...,N} \nonumber \\
			& \hspace{2mm} \vdots \nonumber \\
			C_{var_N} &= \suchthat{X_i^{\alpha_{i+(N-1)}}}{i=1,...,N} 
		\end{align} 	

		\noindent All the individual measurement outcomes are in the $X$ observable, and thus inherit the group structure $(K,\oplus,0)$. The joint outcomes for the measurement in each context $C$ carry the structure of $K^C$, isomorphic to $K^N$ for all $C \in \mathcal{M}$. The empirical model $(e_C)_{C \in \mathcal{M}}$ can then be written as follows, where \inlineQuote{probabilities} are valued in the semiring $R$ of scalars:
		\begin{equation}
			\begin{array}{c|c|c|c}
			& \underline{b} \in K^N \text{ s.t. } \oplus_{i=1}^N b_i =_K 0 
			& \underline{b} \in K^N \text{ s.t. } \oplus_{i=1}^N b_i =_K a
			& \underline{b} \in K^N \text{ s.t. } \oplus_{i=1}^N b_i \neq_K 0,a \vspace{1mm}\\
			\hline &&&\\
			C_{control} 	&	1/d^{(N-1)}		&	0 			& 0 	\\  
			C_{var_1} 		&	0				&	1/d^{(N-1)} & 0 	\\  
			C_{var_2} 		&	0				&	1/d^{(N-1)} & 0 	\\  
			\vdots			&	\vdots			& 	\vdots		& \vdots\\ 
			C_{var_N} 		&	0				&	1/d^{(N-1)} & 0
			\end{array}
		\end{equation}

		\begin{theorem} Assume that the semiring $R$ admits a morphism $p : R \rightarrow \mathbb{B}$ sending $1/d \in R$ to $p(1/d) = 1 \in \mathbb{B}$. Then the empirical model $(e_C)_{C \in \mathcal{M}}$ is $\AvN{\integers}{K}$ if and only if Equation \ref{eqn_UnsolvableEquationAvNSection} admits no solution in $K$.
		\end{theorem}
		\begin{proof} If the semiring $R$ admits a morphism of semirings $p : R \rightarrow \mathbb{B}$ sending $1/d \in R$ to $p(1/d) = 1 \in \mathbb{B}$, then we can associate to $(e_C)_{C \in \mathcal{M}}$ the following possibilistic empirical model $(p \circ e_C)_{C \in \mathcal{M}}$: 
		\begin{equation}
			\begin{array}{c|c|c|c}
			& \underline{b} \in K^N \text{ s.t. } \oplus_{i=1}^N b_i =_K 0 
			& \underline{b} \in K^N \text{ s.t. } \oplus_{i=1}^N b_i =_K a
			& \underline{b} \in K^N \text{ s.t. } \oplus_{i=1}^N b_i \neq_K 0,a \vspace{1mm}\\
			\hline &&&\\
			C_{control} 	&	1				&	0 			& 0 	\\  
			C_{var_1} 		&	0				&	1 			& 0 	\\  
			C_{var_2} 		&	0				&	1 			& 0 	\\  
			\vdots			&	\vdots			& 	\vdots		& \vdots\\ 
			C_{var_N} 		&	0				&	1 			& 0
			\end{array}
		\end{equation}

		\noindent This possibilistic empirical model has the following support subpresheaf $\supportSubpresheafSym \subseteq \sheafOfEventsSym$:
		\begin{align}
			\supportSubpresheaf{C_{control}} &= \suchthat{(b_i^{0})_{i=1}^N \in K^{C_{control}}}{\oplus_{i=1}^N b_i^{0} =_K 0} \nonumber \\
			\supportSubpresheaf{C_{var_1}} &= \suchthat{(b_i^{\alpha_{i}})_{i=1}^N \in K^{C_{var_1}}}{\oplus_{i=1}^N b_i^{\alpha_{i}} =_K a} \nonumber \\
			\supportSubpresheaf{C_{var_2}} &= \suchthat{(b_i^{\alpha_{i+1}})_{i=1}^N \in K^{C_{var_2}}}{\oplus_{i=1}^N b_i^{\alpha_{i+1}} =_K a}  \nonumber \\
			\supportSubpresheaf{C_{var_3}} &= \suchthat{(b_i^{\alpha_{i+2}})_{i=1}^N \in K^{C_{var_3}}}{\oplus_{i=1}^N b_i^{\alpha_{i+2}} =_K a}  \nonumber \\
			& \hspace{2mm} \vdots \nonumber \\
			\supportSubpresheaf{C_{var_N}} &= \suchthat{(b_i^{\alpha_{i+(N-1)}})_{i=1}^N \in K^{C_{var_N}}}{\oplus_{i=1}^N b_i^{\alpha_{i+(N-1)}} =_K a} 
		\end{align} 

		\noindent Amongst the (many) equations in $\RlinearTheory{R}{\supportSubpresheafSym}$ we can find the following $N+1$ equations:
		\begin{align}
			\bigoplus_{m \in C_{control}} \hspace{-3.5mm}s_m &= 0 \text{, safistied by all } s \in \supportSubpresheaf{C_{control}} \nonumber\\
			\bigoplus_{m \in C_{var_1}} \hspace{-2mm}s_m &= a \text{, safistied by all } s \in \supportSubpresheaf{C_{var_1}} \nonumber\\
			& \hspace{1mm} \vdots \nonumber\\
			\bigoplus_{m \in C_{var_N}} \hspace{-2mm}s_m &= a \text{, safistied by all } s \in \supportSubpresheaf{C_{var_N}} 
		\end{align} 

		\noindent Any global assignment $s = (b^0,b^{a_1},...,b^{a_M}): \mathcal{X} \rightarrow K$ satisfying all equations in $\RlinearTheory{R}{\supportSubpresheafSym}$ would in particular satisfy the $N+1$ equations above, and hence provide a solution $y_r := b^{a_r}$ in $K$ to Equation \ref{eqn_UnsolvableEquationAvNSection}. We conclude that, if Equation \ref{eqn_UnsolvableEquationAvNSection} has no solution in $K$, then the empirical model is $\AvN{\integers}{K}$ (where we used the fact that every abelian group is a $\integers$-module). On the other hand, we have seen in the first part of this paper that the empirical model is non-contextual if a solution in $K$ exists, and hence it cannot be $\AvN{\integers}{K}$ in that case.
		\end{proof}

		\noindent The analysis above can be straightforwardly generalised to the case where Equation \ref{eqn_UnsolvableEquationAvNSection} is substituted by a system $\mathbb{S}$ of $\integers$-module equations valued in $K$. As a corollary, operational Mermin non-locality arguments provide a new, infinite family of quantum realisable All-vs-Nothing empirical models.

		\begin{corollary} The operational Mermin non-locality arguments provide an infinite family of quantum realisable $\AvN{\integers}{K}$ empirical models $(e_C^{(K,\mathbb{S})})_{C \in \mathcal{M}^{(K,\mathbb{S})}}$, indexed by all finite abelian groups $K$ and all finite systems $\mathbb{S}$ of $\integers$-module equations valued in $K$ and admitting no solution in $K$. Furthermore, all $\AvN{\integers}{K}$ arguments for some fixed $K$ are equivalently $\AvN{\integersMod{q}}{K}$ for any natural $q$ divisible by the exponent of $K$. In particular, there are operational Mermin non-locality arguments providing $\AvNring{\integersMod{p}}$ models for all primes $p$.
		\end{corollary}

		\noindent Finally, the infinite family of All-vs-Nothing empirical models provided by the previous corollary contains examples of quantum realisable $\AvN{\integersMod{p}}{K}$ models (in fact $\AvNring{\integersMod{p}}$) which are not $\AvN{\integersMod{n}}{K'}$ for any integer $n$ coprime with $p$ and any $\integersMod{n}$-module $K'$, showing that the hierarchy of quantum realisable $\AvN{R}{K}$ models over the quotient rings $R$ of $\integers$ does not collapse.

		\begin{theorem} For each prime $p \geq 3$, there is a quantum realisable $\AvN{\integersMod{p}}{K}$ empirical model (and hence also $\AvN{\integers}{K}$; here $K = \integersMod{p}$) which is not $\AvN{\integersMod{n}}{K'}$ for any natural $n$ coprime with $p$ and any non-trivial abelian group $K'$ with exponent dividing $n$ (in particular, it is not $\AvNring{\integersMod{q}}$ for any prime $q \neq p$).
		\end{theorem}
		\begin{proof}
			Consider the following $\integers$-module equation for any prime $p \geq 3$, which has no solution in the finite abelian group\footnote{Which is also a $\integersMod{p}$-module, hence a $\integers$-module, and a $\integersMod{p}$-vector space.} $K = \integersMod{p}$:
			\begin{equation}
				p y = 1
			\end{equation}

			\noindent The corresponding\footnote{With group-system pair $(K,\{py=1\})$.} operational Mermin non-locality argument gives a quantum realisable empirical model $(e_C)_{C \in \mathcal{M}}$ which is $\AvNring{\integersMod{p}}$, and therefore also $\AvN{\integers}{\integersMod{p}}$; call its support subpresheaf $\supportSubpresheafSym$, and the associated $R$-linear theory $\RlinearTheory{\integersMod{p}}{\supportSubpresheafSym}$.

			If $K'$ is any finite abelian group with exponent dividing $1 < n < p$, then $K' \isom \prod_{l=1}^L \integersMod{p_l}$ for some primes $p_1,...,p_L$ all distinct from $p$. The $\integers$-module equation $p y = 1$ then has a solution mod $y_l$ for all $l$, and hence a solution $y$ in $K'$; consider the global assignment $s:\mathcal{X} \rightarrow K'$ given by $s_0 := 0$ and $s_{a_r} := y$. 
			Since $\integersMod{p}$ is a field, any two equations $\phi,\phi' \in \supportSubpresheaf{C}$ for some $C \in \mathcal{M}$ are non-zero multiples of each other:
			\begin{align}
				\RlinearTheory{\integersMod{p}}{\supportSubpresheaf{C_{control}}} &= \suchthat{\bigoplus_{m \in C_{control}} \hspace{-3.5mm}u s_m = 0 \hspace{2mm}}{u \in \integersMod{p}^\times} \nonumber \\
				\RlinearTheory{\integersMod{p}}{\supportSubpresheaf{C_{var_1}}} &= \suchthat{\hspace{1.2mm}\bigoplus_{m \in C_{var_1}} \hspace{-2mm}u s_m = ua}{u \in \integersMod{p}^\times} \nonumber \\
				& \hspace{1mm} \vdots \nonumber \\
				\RlinearTheory{\integersMod{p}}{\supportSubpresheaf{C_{var_N}}} &= \suchthat{\hspace{0.75mm}\bigoplus_{m \in C_{var_N}} \hspace{-2.3mm}u s_m = ua}{u \in \integersMod{p}^\times} 
			\end{align}

			\noindent In order to compare to an All-vs-Nothing class with respect to a ring $\integersMod{n}$ incompatible with $\integersMod{p}$ (no ring homomorphisms exist between the two), we lift $\RlinearTheory{\integersMod{p}}{\supportSubpresheafSym}$ to the integers by seeing $K = \integersMod{p}$ as a $\integers$-module, and obtain a new set  $\RlinearTheory{\integers}{\supportSubpresheafSym}$ of equations, equivalent to $\RlinearTheory{\integersMod{p}}{\supportSubpresheafSym}$ for assignments over $\integersMod{p}$-modules (like $K=\integersMod{p}$):
			\begin{align}
				\RlinearTheory{\integers}{\supportSubpresheaf{C_{control}}} &= \suchthat{\bigoplus_{m \in C_{control}} \hspace{-3.5mm}u s_m = 0\hspace{2mm}}{u \in \integers \backslash p\integers} \nonumber \\
				\RlinearTheory{\integers}{\supportSubpresheaf{C_{var_1}}} &= \suchthat{\hspace{1.2mm}\bigoplus_{m \in C_{var_1}} \hspace{-2mm}u s_m = ua}{u \in \integers \backslash p\integers} \nonumber \\
				& \hspace{1mm} \vdots \nonumber \\
				\RlinearTheory{\integers}{\supportSubpresheaf{C_{var_N}}} &= \suchthat{\hspace{0.75mm}\bigoplus_{m \in C_{var_N}} \hspace{-2.3mm}u s_m = ua}{u \in \integers \backslash p\integers} 
			\end{align}

			\noindent Over the $\integers$-module $K'$, the set $\RlinearTheory{\integers}{\supportSubpresheafSym}$ always admits a global assignment $s: \mathcal{X} \rightarrow K'$ defined above, and hence the empirical model $(e_C)_{C \in \mathcal{M}}$ is not $\AvN{\integers}{K'}$, and in particular not $\AvN{\integersMod{n}}{K'}$.
		\end{proof}

\section{Conclusions}
	
	We have shown that the operational Mermin non-locality arguments provide a family of quantum realisable All-vs-Nothing models indexed by all pairs $(K,\mathbb{S})$ of finite groups $K$ and systems $\mathbb{S}$ of $K$-valued $\integers$-module equations unsolvable over $K$. In particular, they can be used to show that the hierarchy of quantum realisable All-vs-Nothing models over finite fields does not collapse. Because All-vs-Nothing models are maximally contextual, i.e. lie on a face of the no-signalling polytope, operational Mermin non-locality arguments provide a concrete resource for quantum information and security. An interesting future development would be their application to the design of a larger family of quantum secret sharing protocols, generalising the existing HBB CQ. Another open question is the possible generalisation to states different from the GHZ: while a wide generalisation seems unlikely, the W states show enough structural similarities to be promising candidates.

\subparagraph*{Acknowledgements}
The author would like to thank Samson Abramsky, William Zeng, Vladimir Zamdzhiev, Kohei Kishida, Rui Soares Barbosa,, Raymond Lal and Shane Mansfield for comments, suggestions and useful discussions, as well as Sukrita Chatterji and Nicol\`o Chiappori for their support. Funding from EPSRC and Trinity College is gratefully acknowledged.

\bibliography{bibliography/CategoryTheory,bibliography/CategoricalQM,bibliography/NonLocalityContextuality,bibliography/QuantumComputing,bibliography/ClassicalMechanics,bibliography/LogicComputation,bibliography/Gravitation,bibliography/QFT,bibliography/StatisticalPhysics,bibliography/Misc,bibliography/StefanoGogioso}

\appendix

\section{Gaussian elimination in $S^1$}
\label{section_AppendixGaussianEliminationS1}

	\noindent If $K$ is an abelian group, by a (finite) \textbf{system $\mathbb{S}$ of $K$-valued $\integers$-module equations} we mean a finite family of equations in the following form, with $n_r^s$ integers coefficients, $a^s$ elements of $K$ and $y_r$ the unknowns:
	\begin{equation}\label{eqn_systemAlgExtAppendix}
	\begin{cases}
		\bigoplus_{r=1}^{M} n^1_r \, y_r = a^1 \\
		\vdots  \hspace{1cm}\\
		\bigoplus_{r=1}^{M} n^S_r \, y_r = a^S 
	\end{cases}
	\end{equation}  

	\noindent We will say that a system in the form of \ref{eqn_systemAlgExtAppendix} is \textbf{consistent} if, letting $\underline{n}^s := (n^s_1,...,n^s_M)$ be the row vectors in $\integers^M$, the following holds:
	\begin{equation}
		\bigoplus_{j=1}^{J} c_j \cdot \underline{n}^{s_j} =_{\integers^M} \underline{0} \implies \bigoplus_{j=1}^{J} c_j \cdot a^{s_j} =_{K} 0
	\end{equation}
	\noindent for all naturals $J \geq 1$ and all $(c_1,...,c_J) \in \integers^J$. If $P$ is an abelian group such that $K \leq P$, then by a \textbf{solution in $P$} of a system $\mathbb{S}$ in the form of \ref{eqn_systemAlgExtAppendix} we mean a family $(\beta_r)_{r=1,...,M}$ of elements of $P$ such that setting $y_r := \beta_r$ satisfies the system.

	While all $K$-valued systems with solutions in some super-group of $K$ must necessarily be consistent, the converse is not true in general: given a super-group $P$ of $K$ there may be consistent systems with no solutions in $P$. Certainly if $P$ is finite then at least one such system exists, because of the finite exponent, and certainly if $P=\rationals$ then no such system exists; in fact, every divisible torsion-free abelian group $P$ is canonically a $\rationals$-vector space, and thus every consistent system of $\integers$-modules equations  (and, in fact, of $\rationals$-vector space equations) valued in $P$ has solutions in $P$. Unfortunately, while $S^1$ is divisible it is not torsion-free, and in particular not a $\rationals$-vector space, so the reasoning above doesn't apply. However, we can show that Gaussian elimination can still be performed in $S^1$, and thus that every consistent system of $\integers$-modules equations valued in $S^1$ has solutions in $S^1$.\footnote{However, uniqueness of solution doesn't in general hold for systems with linearly independent row vectors.}

	Consider a system $\mathbb{S}$ in the form of \ref{eqn_systemAlgExtAppendix}. The Gaussian elimination algorithm in a $\rationals$-vector space $V$ can be formulated in terms of the following two fundamental operations:\footnote{We continue to use $\oplus$ and $\ominus$ to denote addition and subtraction of vectors.}
	\begin{enumerate}
	\item[(a)] multiply a row by a non-zero rational $\frac{p}{q} \in \rationals$:
		\begin{equation}
			\left[\bigoplus_{r=1}^{M} n^s_r \, x_r = h^s \right] \mapsto \left[\bigoplus_{r=1}^{M} \frac{p \, n^s_r}{q} \, x_r = \frac{p}{q}h^s \right]
		\end{equation}
	\item[(b)] subtract a rational multiple $\frac{p}{q} \in \rationals$ of a row from another row:
		\begin{equation}
			\left[\bigoplus_{r=1}^{M} n^s_r \, x_r = h^s \right] \mapsto \left[\bigoplus_{r=1}^{M} (n^s_r \ominus \frac{p}{q}\,n^t_r) \, x_r = (h^s \ominus \frac{p}{q}h^t) \right]
		\end{equation}
	\end{enumerate}

	\noindent The correctness of the algorithm is based on the following functionality:
	\begin{enumerate}
	\item[(a)] $\frac{p}{q}\underline{x}=\frac{p}{q}\underline{y} \implies \underline{x} = \underline{y}$ for all non-zero rationals $\frac{p}{q} \in \rationals$ and all vectors $\underline{x},\underline{y} \in V$.
	\item[(b)] $\underline{x}-\frac{p}{q}\underline{z} = \underline{y}-\frac{p}{q}\underline{z} \implies \underline{x}=\underline{y}$ for all rationals $\frac{p}{q} \in \rationals$ and all vectors $\underline{x},\underline{y},\underline{z} \in V$.
	\end{enumerate}

	\noindent Because $S^1$ is not torsion-free, the division by a non-zero natural $n$ is not in general single-valued (in fact, it is exactly $n$-valued), and the functionality required by the usual correctness proof of Gaussian elimination for systems of $\rationals$-valued equations fails. However, one can devise a simple multiple-valued extension that fits the purpose in the $S^1$ case. Rewrite the system from \ref{eqn_systemAlgExtAppendix} as follows, where instead of equations one consider the more general non-deterministic case of intersection of sets:
	\begin{equation}\label{eqn_systemAlgExtAppendixMultival}
	\begin{cases}
		\bigoplus_{r=1}^{M} n^1_r \, x_r \cap \{h^1\} \neq \emptyset \\
		\vdots  \hspace{1cm}\\
		\bigoplus_{r=1}^{M} n^S_r \, x_r \cap \{h^S\} \neq \emptyset 
	\end{cases}
	\end{equation}  
	
	\noindent For divisible abelian groups $P$, one can extend Gaussian elimination by using the following non-deterministic variants of the standard operations:
	\begin{enumerate}
	\item[(a')] multiply a row by a non-zero rational $\frac{p}{q} \in \rationals$:
		\begin{equation}
			\left[\bigoplus_{r=1}^{M} n^s_r \, x_r \cap A^s \neq \emptyset \right] \mapsto \left[\bigoplus_{r=1}^{M} \frac{p \, n^s_r}{q} \, x_r \cap \frac{p}{q}A^s \neq \emptyset \right] 
		\end{equation}
	\item[(b')] subtract a rational multiple $\frac{p}{q} \in \rationals$ of a row from another row:
		\begin{equation}
			\left[\bigoplus_{r=1}^{M} n^s_r \, x_r \cap A^s \neq \emptyset \right] \mapsto \left[] \bigoplus_{r=1}^{M} (n^s_r \ominus p\,n^t_r) \, x_r \cap (A^s \ominus \frac{p}{q}A^t) \neq \emptyset \right]
		\end{equation}
	\end{enumerate}
	
	\noindent The definition of the set $\frac{p}{q}A$ is crucial. If $A$ a set of elements of $P$, we let:
	\begin{equation}
	\frac{p}{q}A := \suchthat{p b }{q b\in A}
	\end{equation}

	\noindent Then the following fundamental property holds for all sets $A$ and $B$ of elements of $P$:
	\begin{equation}
		Y \cap \frac{p}{q}A = \suchthat{p b}{p b \in Y \text{ and } q b \in A} = \frac{q}{p} Y \cap A 
	\end{equation}

	\noindent The definition of the set $A \oplus B$ is more straightforward:
	\begin{equation}
	A+B = \suchthat{a \oplus b}{a \in A \text{ and }b \in B}
	\end{equation}

	\noindent As long as $\frac{p}{q} A \neq \emptyset$ for all $A \neq \emptyset$ and all non-zero $\frac{p}{q} \in \rationals$, it is straightforward to see that this non-deterministic Gaussian elimination will result in each $x_r$ being a non-empty set of elements of $P$ such that the set equations from \ref{eqn_systemAlgExtAppendixMultival} will hold. Any abelian divisible group fits the bill, and in particular so does $S^1$.

\end{document}

%% file: modules/packages.tex
\usepackage{graphicx}
\usepackage{amsfonts}
\usepackage{amsmath}  
\usepackage{amsthm}
\usepackage{amssymb}
\usepackage{relsize}
\usepackage{mathtools}
\usepackage[nocompress]{cite}
\usepackage[usenames,dvipsnames]{xcolor}
\usepackage{tikz}
\usepackage{tikzfig}
\usepackage{fullpage}
\usepackage{stmaryrd}

\usetikzlibrary{arrows,shapes,decorations,backgrounds,positioning,circuits.ee.IEC}


%% file: modules/macros.tex
		\newcounter{theorem_c} 
		\numberwithin{theorem_c}{section} 
		\numberwithin{equation}{section} 

		\theoremstyle{plain} 
		\newtheorem{theorem}[theorem_c]{Theorem}
		
		\newtheorem{lemma}[theorem_c]{Lemma}
		\newtheorem{corollary}[theorem_c]{Corollary}
		\newtheorem*{theorem*}{Theorem}
		
		\newtheorem{remark}[theorem_c]{Remark}

		\theoremstyle{definition}

	\newcommand{\inlineQuote}[1]{\textquotedblleft #1\textquotedblright} 

	\newcommand{\Powerset}[1]{\mathcal{P}(#1)} 
	\newcommand{\integers}{\mathbb{Z}} 
	\newcommand{\rationals}{\mathbb{Q}} 
	\newcommand{\reals}{\mathbb{R}} 
	\newcommand{\complexs}{\mathbb{C}} 
	\newcommand{\integersMod}[1]{\mathbb{Z}_{#1}} 
	\newcommand{\modclass}[2]{#1 \; (\text{mod } #2)} 
	\newcommand{\restrict}[2]{\left. #1 \right\vert_{#2}} 
	\newcommand{\domain}[1]{\operatorname{dom}#1} 
	\newcommand{\support}[1]{\operatorname{supp}(#1)} 
	\newcommand{\inject}{\hookrightarrow} 

	\newcommand{\eqdef}{:=} 
	\newcommand{\suchthat}[2]{\left\{#1 \: \middle\vert \: #2\right\}} 

		\newcommand{\ket}[1]{\vert #1 \rangle} 
		\newcommand{\bra}[1]{\langle #1 \vert} 
		\newcommand{\braket}[2]{\langle #1 \vert #2 \rangle} 

		\newcommand{\im}[1]{\operatorname{im}#1}
		\newcommand{\LinExt}[1]{\operatorname{LinExt}[#1]}

		\newcommand{\ltwoSym}{\ell^2} 
		\newcommand{\ltwo}[1]{\ltwoSym[#1]} 

		\newcommand{\SpaceH}{\mathcal{H}}

		\newcommand{\SpaceG}{\mathcal{G}}

		\newcommand{\isom}{\cong} 
		\newcommand{\id}[1]{id_{#1}} 

		\newcommand{\tensor}{\otimes} 

		\newcommand{\SetCategory}{\operatorname{Set}} 

		\newcommand{\fdHilbCategory}{\operatorname{fdHilb}} 
		\newcommand{\fPInjCategory}{\operatorname{fPInj}} 

		\newcommand{\OpCategory}[1]{#1^{\operatorname{op}}} 
		\newcommand{\CPMCategory}[1]{\operatorname{CPM}[#1]} 

	\newcommand{\sheafOfEventsSym}{\mathcal{E}} 
	\newcommand{\sheafOfEvents}[1]{\sheafOfEventsSym[#1]} 
	\newcommand{\restrictionMap}[2]{\operatorname{res}^{#1}_{#2}} 
	\newcommand{\distributionFunctorSym}[1]{\mathcal{D}_{#1}} 
	\newcommand{\distributionFunctor}[2]{\distributionFunctorSym{#1}[#2]} 
	\newcommand{\presheafOfDistributionsSym}[1]{\distributionFunctorSym{#1}\sheafOfEventsSym} 
	\newcommand{\presheafOfDistributions}[2]{\presheafOfDistributionsSym{#1}[#2]} 
	\newcommand{\supportSubpresheafSym}{\mathcal{S}} 
	\newcommand{\supportSubpresheaf}[1]{\supportSubpresheafSym[#1]} 

	\newcommand{\Xcolour}{Red}
	\newcommand{\XdotSym}{\hbox{\input{modules/symbols/XdotSym.tex}}\!} 
	\newcommand{\kpoints}[1]{K_{#1}}

	\newcommand{\Zcolour}{YellowGreen}
	\newcommand{\ZdotSym}{\hbox{\input{modules/symbols/ZdotSym.tex}}\!} 
	\newcommand{\Zkpoints}{\kpoints{\hbox{\input{modules/symbols/ZdotSym.tex}}\!}}

	\newcommand{\Xaltcolour}{Purple}
	
	\newcommand{\Zaltcolour}{Cyan}
	
	\newcommand{\Dcolour}{black}
	\newcommand{\DdotSym}{\hbox{\input{modules/symbols/DdotSym.tex}}\!} 

	\tikzset{->-/.style={decoration={markings,mark=at position #1 with {\arrow{>}}},postaction={decorate}}}
	\tikzset{-<-/.style={decoration={markings,mark=at position #1 with {\arrow{<}}},postaction={decorate}}}

%% file: modules/symbols/XdotSym.tex
\begin{tikzpicture} [scale=1,transform shape] 

\def\deltax{0.3} 
\def\deltay{0.5} 


\node [dot, fill=\Xcolour] (mult) at (0,0) {};

\end{tikzpicture}

%% file: modules/symbols/ZdotSym.tex
\begin{tikzpicture} [scale=1,transform shape] 

\def\deltax{0.3} 
\def\deltay{0.5} 


\node [dot, fill=\Zcolour] (mult) at (0,0) {};

\end{tikzpicture}

%% file: modules/symbols/DdotSym.tex
\begin{tikzpicture} [scale=1,transform shape] 

\def\deltax{0.3} 
\def\deltay{0.5} 


\node [dot, fill=\Dcolour] (mult) at (0,0) {};

\end{tikzpicture}

%% file: modules/tikzstyles.tex
\tikzstyle{env}=[copoint,regular polygon rotate=0,minimum width=0.2cm, fill=black]

\tikzstyle{probs}=[shape=semicircle,fill=white,draw=black,shape border rotate=180,minimum width=1.2cm]

\tikzstyle{every picture}=[baseline=-0.25em,scale=0.5]
\tikzstyle{dotpic}=[] 
\tikzstyle{diredges}=[every to/.style={diredge}]
\tikzstyle{math matrix}=[matrix of math nodes,left delimiter=(,right delimiter=),inner sep=2pt,column sep=1em,row sep=0.5em,nodes={inner sep=0pt},text height=1.5ex, text depth=0.25ex]

\tikzstyle{inline text}=[text height=1.5ex, text depth=0.25ex,yshift=0.5mm]
\tikzstyle{label}=[font=\footnotesize,text height=1.5ex, text depth=0.25ex,yshift=0.5mm]
\tikzstyle{left label}=[label,anchor=east,xshift=1.5mm]
\tikzstyle{right label}=[label,anchor=west,xshift=-1.5mm]

\tikzstyle{braceedge}=[decorate,decoration={brace,amplitude=2mm,raise=-1mm}]
\tikzstyle{small braceedge}=[decorate,decoration={brace,amplitude=1mm,raise=-1mm}]

\tikzstyle{doubled}=[line width=1.6pt] 
\tikzstyle{boldedge}=[doubled,shorten <=-0.17mm,shorten >=-0.17mm]
\tikzstyle{boldedgegray}=[doubled,gray,shorten <=-0.17mm,shorten >=-0.17mm]

\tikzstyle{semidoubled}=[line width=1.4pt] 
\tikzstyle{semiboldedgegray}=[semidoubled,gray,shorten <=-0.17mm,shorten >=-0.17mm]

\tikzstyle{boldedgedashed}=[very thick,dashed,shorten <=-0.17mm,shorten >=-0.17mm]
\tikzstyle{vboldedgedashed}=[doubled,dashed,shorten <=-0.17mm,shorten >=-0.17mm]
\tikzstyle{left hook arrow}=[left hook-latex]
\tikzstyle{right hook arrow}=[right hook-latex]
\tikzstyle{sembracket}=[line width=0.5pt,shorten <=-0.07mm,shorten >=-0.07mm]

\tikzstyle{causal edge}=[->,thick,gray]
\tikzstyle{causal nondir}=[thick,gray]
\tikzstyle{timeline}=[thick,gray, dashed]

\tikzstyle{cedge}=[<->,thick,gray!70!white]

\tikzstyle{empty diagram}=[draw=gray!40!white,dashed,shape=rectangle,minimum width=1cm,minimum height=1cm]
\tikzstyle{empty diagram small}=[draw=gray!50!white,dashed,shape=rectangle,minimum width=0.6cm,minimum height=0.5cm]

\tikzstyle{dot}=[inner sep=0mm,minimum width=3mm,minimum height=3mm,draw,shape=circle,text depth=-0.1mm]
\tikzstyle{ddot}=[inner sep=0mm, doubled, minimum width=3.5mm,minimum height=3.5mm,draw,shape=circle]

\tikzstyle{black dot}=[dot,fill=black]
\tikzstyle{white dot}=[dot,fill=white,,text depth=-0.2mm]
\tikzstyle{green dot}=[white dot] 
\tikzstyle{gray dot}=[dot,fill=gray!40!white,,text depth=-0.2mm]
\tikzstyle{red dot}=[gray dot] 

\tikzstyle{black ddot}=[ddot,fill=black]
\tikzstyle{white ddot}=[ddot,fill=white]
\tikzstyle{gray ddot}=[ddot,fill=gray!40!white]

\tikzstyle{gray edge}=[gray!40!white]

\tikzstyle{small dot}=[inner sep=0.5mm,minimum width=0pt,minimum height=0pt,draw,shape=circle]

\tikzstyle{small black dot}=[small dot,fill=black]
\tikzstyle{small white dot}=[small dot,fill=white]
\tikzstyle{small gray dot}=[small dot,fill=gray!40!white]

\tikzstyle{causal dot}=[inner sep=0.4mm,minimum width=0pt,minimum height=0pt,draw=white,shape=circle,fill=gray!40!white]

\tikzstyle{phase dimensions}=[minimum size=5mm,font=\footnotesize,rectangle,rounded corners=2.5mm,inner sep=0.2mm,outer sep=-2mm,text height=1ex, text depth=0.25ex, yshift=0.5mm]
\tikzstyle{dphase dimensions}=[phase dimensions]

\tikzstyle{phase dot}=[dot,phase dimensions]

\tikzstyle{white phase dot}=[dot,fill=white,phase dimensions]
\tikzstyle{white phase ddot}=[ddot,fill=white,dphase dimensions]

\tikzstyle{white rect ddot}=[draw=black,fill=white,doubled,minimum size=5mm,font=\footnotesize,rectangle,rounded corners=2.5mm,inner sep=0.2mm]
\tikzstyle{gray rect ddot}=[draw=black,fill=gray!40!white,doubled,minimum size=6mm,font=\footnotesize,rectangle,rounded corners=3mm]

\tikzstyle{gray phase dot}=[dot,fill=gray!40!white,phase dimensions]
\tikzstyle{gray phase ddot}=[ddot,fill=gray!40!white,dphase dimensions]
\tikzstyle{grey phase dot}=[gray phase dot]
\tikzstyle{grey phase ddot}=[gray phase ddot]

\tikzstyle{cnot}=[fill=white,shape=circle,inner sep=-1.4pt]
\tikzstyle{hadamard}=[square box,inner sep=0 pt,font=\footnotesize,minimum height=4mm,minimum width=4mm]
\tikzstyle{dhadamard}=[hadamard,doubled]
\tikzstyle{antipode}=[white dot,inner sep=0.3mm,font=\footnotesize]

\tikzstyle{scalar}=[diamond,draw,inner sep=0.5pt,font=\small]
\tikzstyle{dscalar}=[diamond,doubled, draw,inner sep=0.5pt,font=\small]

\tikzstyle{small box}=[rectangle,inline text,fill=white,draw,minimum height=5mm,yshift=-0.5mm,minimum width=5mm,font=\small]
\tikzstyle{small gray box}=[small box,fill=gray!30]
\tikzstyle{medium box}=[rectangle,inline text,fill=white,draw,minimum height=5mm,yshift=-0.5mm,minimum width=10mm,font=\small]
\tikzstyle{square box}=[small box] 
\tikzstyle{medium gray box}=[small box,fill=gray!30]
\tikzstyle{semilarge box}=[rectangle,inline text,fill=white,draw,minimum height=5mm,yshift=-0.5mm,minimum width=12.5mm,font=\small]
\tikzstyle{large box}=[rectangle,inline text,fill=white,draw,minimum height=5mm,yshift=-0.5mm,minimum width=15mm,font=\small]
\tikzstyle{large gray box}=[small box,fill=gray!30]

\tikzstyle{gray square point}=[small box,fill=gray!50]

\tikzstyle{dphase box white}=[dbox]
\tikzstyle{dphase box gray}=[dbox,fill=gray!50!white]

\tikzstyle{point}=[regular polygon,regular polygon sides=3,draw,scale=0.75,inner sep=-0.5pt,minimum width=9mm,fill=white,regular polygon rotate=180]
\tikzstyle{copoint}=[regular polygon,regular polygon sides=3,draw,scale=0.75,inner sep=-0.5pt,minimum width=9mm,fill=white]
\tikzstyle{dpoint}=[point,doubled]
\tikzstyle{dcopoint}=[copoint,doubled]

\tikzstyle{wide copoint}=[fill=white,draw,shape=isosceles triangle,shape border rotate=90,isosceles triangle stretches=true,inner sep=0pt,minimum width=1.5cm,minimum height=6.12mm]
\tikzstyle{wide point}=[fill=white,draw,shape=isosceles triangle,shape border rotate=-90,isosceles triangle stretches=true,inner sep=0pt,minimum width=1.5cm,minimum height=6.12mm,yshift=-0.0mm]
\tikzstyle{wide point plus}=[fill=white,draw,shape=isosceles triangle,shape border rotate=-90,isosceles triangle stretches=true,inner sep=0pt,minimum width=1.74cm,minimum height=7mm,yshift=-0.0mm]

\tikzstyle{wide dpoint}=[fill=white,doubled,draw,shape=isosceles triangle,shape border rotate=-90,isosceles triangle stretches=true,inner sep=0pt,minimum width=1.5cm,minimum height=6.12mm,yshift=-0.0mm]

\tikzstyle{tinypoint}=[regular polygon,regular polygon sides=3,draw,scale=0.55,inner sep=-0.15pt,minimum width=6mm,fill=white,regular polygon rotate=180] 

\tikzstyle{white point}=[point]
\tikzstyle{white dpoint}=[dpoint]
\tikzstyle{green point}=[white point] 
\tikzstyle{white copoint}=[copoint]
\tikzstyle{gray point}=[point,fill=gray!40!white]
\tikzstyle{gray dpoint}=[gray point,doubled]
\tikzstyle{red point}=[gray point] 
\tikzstyle{gray copoint}=[copoint,fill=gray!40!white]
\tikzstyle{gray dcopoint}=[gray copoint,doubled]

\tikzstyle{black point}=[point,fill=black]
\tikzstyle{black copoint}=[copoint,fill=black]

\tikzstyle{tiny gray point}=[tinypoint,fill=gray!40!white]

\tikzstyle{diredge}=[->]
\tikzstyle{rdiredge}=[<-]
\tikzstyle{thickdiredge}=[->, very thick]
\tikzstyle{pointer edge}=[->,very thick,gray]
\tikzstyle{pointer edge part}=[very thick,gray]
\tikzstyle{dashed edge}=[dashed]
\tikzstyle{thick dashed edge}=[very thick,dashed]
\tikzstyle{thick gray dashed edge}=[thick dashed edge,gray!40]
\tikzstyle{thick map edge}=[very thick,|->]

\makeatletter
\newcommand{\boxshape}[3]{%
\pgfdeclareshape{#1}{
\inheritsavedanchors[from=rectangle] 
\inheritanchorborder[from=rectangle]
\inheritanchor[from=rectangle]{center}
\inheritanchor[from=rectangle]{north}
\inheritanchor[from=rectangle]{south}
\inheritanchor[from=rectangle]{west}
\inheritanchor[from=rectangle]{east}
\backgroundpath{
\southwest \pgf@xa=\pgf@x \pgf@ya=\pgf@y
\northeast \pgf@xb=\pgf@x \pgf@yb=\pgf@y

\@tempdima=#2
\@tempdimb=#3

\pgfpathmoveto{\pgfpoint{\pgf@xa - 5pt + \@tempdima}{\pgf@ya}}
\pgfpathlineto{\pgfpoint{\pgf@xa - 5pt - \@tempdima}{\pgf@yb}}
\pgfpathlineto{\pgfpoint{\pgf@xb + 5pt + \@tempdimb}{\pgf@yb}}
\pgfpathlineto{\pgfpoint{\pgf@xb + 5pt - \@tempdimb}{\pgf@ya}}
\pgfpathlineto{\pgfpoint{\pgf@xa - 5pt + \@tempdima}{\pgf@ya}}
\pgfpathclose
}
}}

\boxshape{NEbox}{0pt}{5pt}
\boxshape{SEbox}{0pt}{-5pt}
\boxshape{NWbox}{5pt}{0pt}
\boxshape{SWbox}{-5pt}{0pt}
\boxshape{EBox}{-3pt}{3pt}
\boxshape{WBox}{3pt}{-3pt}
\makeatother

\tikzstyle{cloud}=[shape=cloud,draw,minimum width=1.5cm,minimum height=1.5cm]

\tikzstyle{map}=[draw,shape=NEbox,inner sep=2pt,minimum height=6mm,fill=white]
\tikzstyle{dashedmap}=[draw,dashed,shape=NEbox,inner sep=2pt,minimum height=6mm,fill=white]
\tikzstyle{mapdag}=[draw,shape=SEbox,inner sep=2pt,minimum height=6mm,fill=white]
\tikzstyle{mapadj}=[draw,shape=SEbox,inner sep=2pt,minimum height=6mm,fill=white]
\tikzstyle{maptrans}=[draw,shape=SWbox,inner sep=2pt,minimum height=6mm,fill=white]
\tikzstyle{mapconj}=[draw,shape=NWbox,inner sep=2pt,minimum height=6mm,fill=white]

\tikzstyle{medium map}=[draw,shape=NEbox,inner sep=2pt,minimum height=6mm,fill=white,minimum width=7mm]
\tikzstyle{medium map dag}=[draw,shape=SEbox,inner sep=2pt,minimum height=6mm,fill=white,minimum width=7mm]
\tikzstyle{medium map adj}=[draw,shape=SEbox,inner sep=2pt,minimum height=6mm,fill=white,minimum width=7mm]
\tikzstyle{medium map trans}=[draw,shape=SWbox,inner sep=2pt,minimum height=6mm,fill=white,minimum width=7mm]
\tikzstyle{medium map conj}=[draw,shape=NWbox,inner sep=2pt,minimum height=6mm,fill=white,minimum width=7mm]
\tikzstyle{semilarge map}=[draw,shape=NEbox,inner sep=2pt,minimum height=6mm,fill=white,minimum width=9.5mm]
\tikzstyle{semilarge map trans}=[draw,shape=SWbox,inner sep=2pt,minimum height=6mm,fill=white,minimum width=9.5mm]
\tikzstyle{semilarge map adj}=[draw,shape=SEbox,inner sep=2pt,minimum height=6mm,fill=white,minimum width=9.5mm]
\tikzstyle{semilarge map dag}=[draw,shape=SEbox,inner sep=2pt,minimum height=6mm,fill=white,minimum width=9.5mm]
\tikzstyle{semilarge map conj}=[draw,shape=NWbox,inner sep=2pt,minimum height=6mm,fill=white,minimum width=9.5mm]
\tikzstyle{large map}=[draw,shape=NEbox,inner sep=2pt,minimum height=6mm,fill=white,minimum width=12mm]
\tikzstyle{very large map}=[draw,shape=NEbox,inner sep=2pt,minimum height=6mm,fill=white,minimum width=17mm]

\tikzstyle{medium dmap}=[draw,doubled,shape=NEbox,inner sep=2pt,minimum height=6mm,fill=white,minimum width=7mm]
\tikzstyle{medium dmap dag}=[draw,doubled,shape=SEbox,inner sep=2pt,minimum height=6mm,fill=white,minimum width=7mm]
\tikzstyle{medium dmap adj}=[draw,doubled,shape=SEbox,inner sep=2pt,minimum height=6mm,fill=white,minimum width=7mm]
\tikzstyle{medium dmap trans}=[draw,doubled,shape=SWbox,inner sep=2pt,minimum height=6mm,fill=white,minimum width=7mm]
\tikzstyle{medium dmap conj}=[draw,doubled,shape=NWbox,inner sep=2pt,minimum height=6mm,fill=white,minimum width=7mm]
\tikzstyle{semilarge dmap}=[draw,doubled,shape=NEbox,inner sep=2pt,minimum height=6mm,fill=white,minimum width=9.5mm]
\tikzstyle{semilarge dmap trans}=[draw,doubled,shape=SWbox,inner sep=2pt,minimum height=6mm,fill=white,minimum width=9.5mm]
\tikzstyle{semilarge dmap adj}=[draw,doubled,shape=SEbox,inner sep=2pt,minimum height=6mm,fill=white,minimum width=9.5mm]
\tikzstyle{semilarge dmap dag}=[draw,doubled,shape=SEbox,inner sep=2pt,minimum height=6mm,fill=white,minimum width=9.5mm]
\tikzstyle{semilarge dmap conj}=[draw,doubled,shape=NWbox,inner sep=2pt,minimum height=6mm,fill=white,minimum width=9.5mm]
\tikzstyle{large dmap}=[draw,doubled,shape=NEbox,inner sep=2pt,minimum height=6mm,fill=white,minimum width=12mm]
\tikzstyle{large dmap conj}=[draw,doubled,shape=NWbox,inner sep=2pt,minimum height=6mm,fill=white,minimum width=12mm]
\tikzstyle{large dmap trans}=[draw,doubled,shape=SWbox,inner sep=2pt,minimum height=6mm,fill=white,minimum width=12mm]
\tikzstyle{very large dmap}=[draw,doubled,shape=NEbox,inner sep=2pt,minimum height=6mm,fill=white,minimum width=19.5mm]

\tikzstyle{muxbox}=[draw,shape=rectangle,minimum height=3mm,minimum width=3mm,fill=white]
\tikzstyle{dmuxbox}=[muxbox,doubled]

\tikzstyle{box}=[draw,shape=rectangle,inner sep=2pt,minimum height=6mm,minimum width=6mm,fill=white]
\tikzstyle{dbox}=[draw,doubled,shape=rectangle,inner sep=2pt,minimum height=6mm,minimum width=6mm,fill=white]
\tikzstyle{dmap}=[draw,doubled,shape=NEbox,inner sep=2pt,minimum height=6mm,fill=white]
\tikzstyle{dmapdag}=[draw,doubled,shape=SEbox,inner sep=2pt,minimum height=6mm,fill=white]
\tikzstyle{dmapadj}=[draw,doubled,shape=SEbox,inner sep=2pt,minimum height=6mm,fill=white]
\tikzstyle{dmaptrans}=[draw,doubled,shape=SWbox,inner sep=2pt,minimum height=6mm,fill=white]
\tikzstyle{dmapconj}=[draw,doubled,shape=NWbox,inner sep=2pt,minimum height=6mm,fill=white]

\tikzstyle{ddmap}=[draw,doubled,dashed,shape=NEbox,inner sep=2pt,minimum height=6mm,fill=white]
\tikzstyle{ddmapdag}=[draw,doubled,dashed,shape=SEbox,inner sep=2pt,minimum height=6mm,fill=white]
\tikzstyle{ddmapadj}=[draw,doubled,dashed,shape=SEbox,inner sep=2pt,minimum height=6mm,fill=white]
\tikzstyle{ddmaptrans}=[draw,doubled,dashed,shape=SWbox,inner sep=2pt,minimum height=6mm,fill=white]
\tikzstyle{ddmapconj}=[draw,doubled,dashed,shape=NWbox,inner sep=2pt,minimum height=6mm,fill=white]

\boxshape{sNEbox}{0pt}{3pt}
\boxshape{sSEbox}{0pt}{-3pt}
\boxshape{sNWbox}{3pt}{0pt}
\boxshape{sSWbox}{-3pt}{0pt}
\tikzstyle{smap}=[draw,shape=sNEbox,fill=white]
\tikzstyle{smapdag}=[draw,shape=sSEbox,fill=white]
\tikzstyle{smapadj}=[draw,shape=sSEbox,fill=white]
\tikzstyle{smaptrans}=[draw,shape=sSWbox,fill=white]
\tikzstyle{smapconj}=[draw,shape=sNWbox,fill=white]

\tikzstyle{dsmap}=[draw,dashed,shape=sNEbox,fill=white]
\tikzstyle{dsmapdag}=[draw,dashed,shape=sSEbox,fill=white]
\tikzstyle{dsmaptrans}=[draw,dashed,shape=sSWbox,fill=white]
\tikzstyle{dsmapconj}=[draw,dashed,shape=sNWbox,fill=white]

\boxshape{mNEbox}{0pt}{10pt}
\boxshape{mSEbox}{0pt}{-10pt}
\boxshape{mNWbox}{10pt}{0pt}
\boxshape{mSWbox}{-10pt}{0pt}
\tikzstyle{mmap}=[draw,shape=mNEbox]
\tikzstyle{mmapdag}=[draw,shape=mSEbox]
\tikzstyle{mmaptrans}=[draw,shape=mSWbox]
\tikzstyle{mmapconj}=[draw,shape=mNWbox]

\tikzstyle{mmapgray}=[draw,fill=gray!40!white,shape=mNEbox]
\tikzstyle{smapgray}=[draw,fill=gray!40!white,shape=sNEbox]

\makeatletter
\pgfdeclareshape{cornerpoint}{
\inheritsavedanchors[from=rectangle] 
\inheritanchorborder[from=rectangle]
\inheritanchor[from=rectangle]{center}
\inheritanchor[from=rectangle]{north}
\inheritanchor[from=rectangle]{south}
\inheritanchor[from=rectangle]{west}
\inheritanchor[from=rectangle]{east}
\backgroundpath{
\southwest \pgf@xa=\pgf@x \pgf@ya=\pgf@y
\northeast \pgf@xb=\pgf@x \pgf@yb=\pgf@y

\pgfmathsetmacro{\pgf@shorten@left}{\pgfkeysvalueof{/tikz/shorten left}}
\pgfmathsetmacro{\pgf@shorten@right}{\pgfkeysvalueof{/tikz/shorten right}}

\pgfpathmoveto{\pgfpoint{0.5 * (\pgf@xa + \pgf@xb)}{\pgf@ya - 5pt}}
\pgfpathlineto{\pgfpoint{\pgf@xa - 8pt + \pgf@shorten@left}{\pgf@yb - 1.5 * \pgf@shorten@left}}
\pgfpathlineto{\pgfpoint{\pgf@xa - 8pt + \pgf@shorten@left}{\pgf@yb}}
\pgfpathlineto{\pgfpoint{\pgf@xb + 8pt - \pgf@shorten@right}{\pgf@yb}}
\pgfpathlineto{\pgfpoint{\pgf@xb + 8pt - \pgf@shorten@right}{\pgf@yb - 1.5 * \pgf@shorten@right}}
\pgfpathclose
}
}

\pgfdeclareshape{cornercopoint}{
\inheritsavedanchors[from=rectangle] 
\inheritanchorborder[from=rectangle]
\inheritanchor[from=rectangle]{center}
\inheritanchor[from=rectangle]{north}
\inheritanchor[from=rectangle]{south}
\inheritanchor[from=rectangle]{west}
\inheritanchor[from=rectangle]{east}
\backgroundpath{
\southwest \pgf@xa=\pgf@x \pgf@ya=\pgf@y
\northeast \pgf@xb=\pgf@x \pgf@yb=\pgf@y

\pgfmathsetmacro{\pgf@shorten@left}{\pgfkeysvalueof{/tikz/shorten left}}
\pgfmathsetmacro{\pgf@shorten@right}{\pgfkeysvalueof{/tikz/shorten right}}

\pgfpathmoveto{\pgfpoint{0.5 * (\pgf@xa + \pgf@xb)}{\pgf@yb + 5pt}}
\pgfpathlineto{\pgfpoint{\pgf@xa - 8pt + \pgf@shorten@left}{\pgf@ya + 1.5 * \pgf@shorten@left}}
\pgfpathlineto{\pgfpoint{\pgf@xa - 8pt + \pgf@shorten@left}{\pgf@ya}}
\pgfpathlineto{\pgfpoint{\pgf@xb + 8pt - \pgf@shorten@right}{\pgf@ya}}
\pgfpathlineto{\pgfpoint{\pgf@xb + 8pt - \pgf@shorten@right}{\pgf@ya + 1.5 * \pgf@shorten@right}}
\pgfpathclose
}
}

\makeatother

\pgfkeyssetvalue{/tikz/shorten left}{0pt}
\pgfkeyssetvalue{/tikz/shorten right}{0pt}

\tikzstyle{kpoint common}=[draw,fill=white,inner sep=1pt,minimum height=3mm]
\tikzstyle{kpoint}=[shape=cornerpoint,shorten left=5pt,kpoint common]
\tikzstyle{kpoint adjoint}=[shape=cornercopoint,shorten left=5pt,kpoint common]
\tikzstyle{kpoint conjugate}=[shape=cornerpoint,shorten right=5pt,kpoint common]
\tikzstyle{kpoint transpose}=[shape=cornercopoint,shorten right=5pt,kpoint common]
\tikzstyle{kpoint symm}=[shape=cornerpoint,shorten left=5pt,shorten right=5pt,kpoint common]

\tikzstyle{black kpoint}=[shape=cornerpoint,shorten left=5pt,kpoint common,fill=black]
\tikzstyle{black kpoint adjoint}=[shape=cornercopoint,shorten left=5pt,kpoint common,fill=black]

\tikzstyle{kpointdag}=[kpoint adjoint]
\tikzstyle{kpointadj}=[kpoint adjoint]
\tikzstyle{kpointconj}=[kpoint conjugate]
\tikzstyle{kpointtrans}=[kpoint transpose]

\tikzstyle{big kpoint}=[kpoint, minimum width=1.2 cm, minimum height=8mm, inner sep=4pt, text depth=3mm]

\tikzstyle{wide kpoint}=[kpoint, minimum width=1 cm, inner sep=2pt, text depth=-0.7 mm]
\tikzstyle{wide kpointdag}=[kpointdag, minimum width=1 cm, inner sep=2pt, text depth=0.7 mm]
\tikzstyle{wide kpointconj}=[kpointconj, minimum width=1 cm, inner sep=2pt, text depth=-0.7 mm]
\tikzstyle{wide kpointtrans}=[kpointtrans, minimum width=1 cm, inner sep=2pt, text depth=0.7 mm]

\tikzstyle{gray kpoint}=[kpoint,fill=gray!50!white]
\tikzstyle{gray kpointdag}=[kpointdag,fill=gray!50!white]
\tikzstyle{gray kpointadj}=[kpointadj,fill=gray!50!white]
\tikzstyle{gray kpointconj}=[kpointconj,fill=gray!50!white]
\tikzstyle{gray kpointtrans}=[kpointtrans,fill=gray!50!white]

\tikzstyle{gray dkpoint}=[kpoint,fill=gray!50!white,doubled]
\tikzstyle{gray dkpointdag}=[kpointdag,fill=gray!50!white,doubled]
\tikzstyle{gray dkpointadj}=[kpointadj,fill=gray!50!white,doubled]
\tikzstyle{gray dkpointconj}=[kpointconj,fill=gray!50!white,doubled]
\tikzstyle{gray dkpointtrans}=[kpointtrans,fill=gray!50!white,doubled]

\tikzstyle{white label}=[draw,fill=white,rectangle,inner sep=0.7 mm]
\tikzstyle{gray label}=[draw,fill=gray!50!white,rectangle,inner sep=0.7 mm]
\tikzstyle{black label}=[draw,fill=black,rectangle,inner sep=0.7 mm]

\tikzstyle{dkpoint}=[kpoint,doubled]
\tikzstyle{wide dkpoint}=[wide kpoint,doubled]
\tikzstyle{dkpointdag}=[kpoint adjoint,doubled]
\tikzstyle{dkcopoint}=[kpoint adjoint,doubled]
\tikzstyle{dkpointadj}=[kpoint adjoint,doubled]
\tikzstyle{dkpointconj}=[kpoint conjugate,doubled]
\tikzstyle{dkpointtrans}=[kpoint transpose,doubled]

\tikzstyle{kscalar}=[kpoint common, shape=EBox, inner xsep=-1pt, inner ysep=3pt,font=\small]
\tikzstyle{kscalarconj}=[kpoint common, shape=WBox, inner xsep=-1pt, inner ysep=3pt,font=\small]

 \tikzstyle{upground}=[circuit ee IEC,thick,ground,rotate=90,scale=2.5]
 \tikzstyle{downground}=[circuit ee IEC,thick,ground,rotate=-90,scale=2.5]
 \tikzstyle{bigground}=[regular polygon,regular polygon sides=3,draw=gray,scale=0.50,inner sep=-0.5pt,minimum width=10mm,fill=gray]

\tikzstyle{arrs}=[-latex,font=\small,auto]
\tikzstyle{arrow plain}=[arrs]
\tikzstyle{arrow dashed}=[dashed,arrs]
\tikzstyle{arrow bold}=[very thick,arrs]
\tikzstyle{arrow hide}=[draw=white!0,-]
\tikzstyle{arrow reverse}=[latex-]
\tikzstyle{cdnode}=[]

%% file: modules/symbols/timematchSym.tex
\begin{tikzpicture} [scale=1.2,transform shape] 

\def\deltax{0.3} 
\def\deltay{0.5} 

\path[use as bounding box] (-\deltax,-\deltay) rectangle (\deltax,\deltay);

\node (mult_label_inl) at (-\deltax,-\deltay) {};
\node (mult_label_inr) at (+\deltax,-\deltay) {};
\node [dot, fill=\Zcolour] (mult) at (0,0) {};
\node (mult_label_out) at (0,+\deltay) {};
\draw[-] [out=90,in=225](mult_label_inl) to (mult);
\draw[-] [out=90,in=315](mult_label_inr) to (mult);
\draw[-] (mult) to (mult_label_out);

\end{tikzpicture}

%% file: modules/symbols/timematchunitSym.tex
\begin{tikzpicture} [scale=1.2,transform shape] 

\def\deltax{0.3} 
\def\deltay{0.5} 

\path[use as bounding box] (-\deltax,-\deltay) rectangle (\deltax,\deltay);

\node [dot, fill=\Zcolour] (mult) at (0,0) {};
\node (mult_label_out) at (0,+\deltay) {};
\draw[-] (mult) to (mult_label_out);

\end{tikzpicture}

%% file: modules/symbols/timediagSym.tex
\begin{tikzpicture} [scale=1.2,transform shape] 

\def\deltax{0.3} 
\def\deltay{0.5} 

\path[use as bounding box] (-\deltax,-\deltay) rectangle (\deltax,\deltay);

\node (mult_label_outl) at (-\deltax,+\deltay) {};
\node (mult_label_outr) at (+\deltax,+\deltay) {};
\node [dot, fill=\Zcolour] (mult) at (0,0) {};
\node (mult_label_in) at (0,-\deltay) {};
\draw[-] [in=270,out=135] (mult) to (mult_label_outl);
\draw[-] [in=270,out=45] (mult) to (mult_label_outr);
\draw[-] (mult_label_in) to (mult);

\end{tikzpicture}